\documentclass[sigconf]{acmart}

\usepackage{booktabs} % For formal tables

\setcopyright{rightsretained}

\usepackage{geometry}                % See geometry.pdf to learn the layout options. There are lots.
\usepackage{graphicx}
\usepackage{amsmath}
\usepackage{amssymb}
\usepackage{amsthm}
\usepackage[noend]{algpseudocode}
\usepackage{epstopdf}
\usepackage{algorithm}
\usepackage{color}
\usepackage{subcaption}
\usepackage{multirow}
\usepackage{wasysym}
\usepackage{xcolor}
\usepackage{tabu}

\newtheorem{remark}{Remark}

\begin{document}

\title{Joint Data Compression and Caching:\\ Approaching Optimality with Guarantees}

\author{Jian Li$^*$}
\affiliation{
  \institution{College of Information and\\ Computer Sciences \\University of Massachusetts }
  \state{Amherst, MA 01003, USA} 
}
\email{jianli@cs.umass.edu}

\author{Faheem Zafari}
\authornote{Co-primary authors with equal contribution}
\affiliation{
  \institution{Department of Electrical and Electronic Engineering, \\Imperial College London, }
  \state{London SW72AZ, U.K.} 
}
\email{faheem16@imperial.ac.uk}

\author{Don Towsley}
\affiliation{
  \institution{College of Information and\\ Computer Sciences \\University of Massachusetts }
  \state{Amherst, MA 01003, USA} 
}
\email{towsley@cs.umass.edu}

\author{Kin K. Leung}
\affiliation{
  \institution{Department of Electrical and Electronic Engineering, }
  \state{Imperial College London,\\London SW72AZ, U.K.} 
}
\email{kin.leung@imperial.ac.uk}

\author{Ananthram Swami}
\affiliation{
  \institution{U.S. Army Research Laboratory}
  \state{Adelphi, MD 20783 USA} 
}
\email{ananthram.swami.civ@mail.mil}

\begin{abstract}
%\begin{abstract}
We consider the problem of optimally compressing and caching data across a communication network.  Given the data generated at edge nodes and a routing path, our goal is to determine the optimal data compression ratios and caching decisions across the network in order to minimize average latency, which can be shown  to be equivalent to maximizing \emph{the compression and caching gain} under an energy consumption constraint.  We show that this problem is NP-hard in general and the hardness is caused by the caching decision subproblem, while the compression sub-problem is polynomial-time solvable.  We then propose an approximation algorithm that achieves a $(1-1/e)$-approximation solution to the optimum in strongly polynomial time.  We show that our proposed algorithm achieve the  near-optimal performance in synthetic-based evaluations.  In this paper, we consider a tree-structured network as an illustrative example, but our results  easily extend to  general network topology at the expense of more complicated notations.
%\end{abstract}

\end{abstract}
\maketitle

\section{Introduction}\label{sec:intro}

In recent years, with the ever increasing prevalence of edge computing enabled mobile devices and applications, such as social media, weather reports, emails notifications, etc., the demand for data communication has significantly increased.   As bandwidth and power supply associated with mobile devices are limited, efficient data communication is critical.  %These include, but are not limited to, military sites used in tactical networks \cite{spencer16}, sensors used in environment monitoring \cite{al10}, social networks for information sharing \cite{bakshy12,goecks04}, and so on. 

In this paper, we consider a network of nodes, each capable of compressing data to a certain degree and caching a constant amount of data.  A certain set of nodes {generates} real time data and a sink node collects them from these nodes through a fixed path to serve requests for these data.  %{The requests for the data are forwarded node-to-node, along the unique path between the sink and leaf nodes, as long as the data is not found at any node. }
{However, the requests need not reach nodes that generated the data, i.e. request forwarding stops upon reaching a node on the path that caches the requested data.}
 Upon finding the data, it is sent  along the reverse path to the sink node to serve the requests.

While each node can cache data to serve future requests so as to reduce access latency and bandwidth requirement, it incurs additional caching costs \cite{choi12}. Furthermore, data compression reduces the transmission cost at the expense of computation cost \cite{barr06,nazemi16}.  Thus, there is an energy consumption tradeoff among data compression, transmission and caching to reduce latency.  Since bandwidth and energy required for network operation is expensive \cite{nazemi16},  it is critical to efficiently compress, transmit and cache the data to reduce the latency. This introduces the following question, what is the right balance between compression and caching to minimize the total communication latency for a limited energy consumption?

Our primary goal is to minimize the average network latency (delay) due to data transfer across the network, subject to an energy consumption constraint on compression and caching of the data.  This problem is naturally abstracted and motivated by many real world applications, including wireless sensor networks (WSNs) \cite{choi12}, peer-to-peer networks \cite{cohen02}, content distribution networks (CDNs) \cite{borst10, dehghan15,jiansri17}, Information Centric Networks (ICNs) \cite{jacobson09} and so on.  For example, in a hierarchical WSN, the sensors generate data, which can be compressed and forwarded to  the sink node through fixed paths to serve requests generated from outside the network.  These requests can be served from the intermediate nodes along the path that cache the data; if, however, data is not cached on any node along the path, the request can subsequently be forwarded to the edge sensor that generates the requested data.  Similarly, in the ICN, requests for data can be served locally from  intermediate caches placed between the server and origin. Both applications can be mapped into the problem we consider here. 

For these and many other applications, it is natural to assume that  edge nodes in the network generate data which is then compressed and transmitted by all the nodes along the path to the sink node. The sink node receives and serves requests generated outside the communication  network. The intermediate nodes along the path can cache  data to serve requests. However, compression, transmission and caching consume energy, while the {node} power supply is usually limited.  To address this challenge, our main goal is to design a lightweight and efficient  algorithm with provable performance guarantees to minimize average latency.  We make the following contributions in this work:
\begin{itemize}
\item We propose a formal mathematical framework for joint data compression and cache optimization. Specifically, we formulate the problem of  finding the optimal data compression ratios and caching locations that minimizes average delay in serving requests subject to an energy constraint. % maximizes the expected compression and caching gain subject to an energy constraint.%, i.e., the latency reduction achieved by compressing and caching the data in intermediate nodes at the expense of energy consumption. 
\item  We analyze the complexity of the problem and show that it is NP-hard in general. The hardness is caused by data allocation to the caches.%. We show that the problem is NP-hard in general, and the hardness is caused by the caching allocation. 
\item We propose polynomial time solvable algorithms for the formulated problem. Since the original optimization problem is NP-hard and non-convex, we first relax the constraints and show that the relaxed problem can be  transformed into an equivalent convex optimization problem that can be solved in polynomial time.  We further show that combining this solution with greedy caching allocation achieves a solution with $1/2$-approximation to the optimum.  Moreover, we construct a polynomial-time ($1-1/e$) approximation algorithm for the problem. 
\item We conduct extensive simulations using  synthetic based network topologies and compare our proposed algorithm with benchmark techniques.  Our results show that the proposed algorithm achieves the near-optimal performance, and significantly outperforms benchmarks Genetic Algorithm \cite{deb02}, Bonmin \cite{bonami08}, and NOMAD \cite{le11} by obtaining a feasible solution in a shorter time for various network topologies.% with respect to robustness in achieving solutions for various  network topologies.
\end{itemize}

The rest of the paper is organized as follows: We discuss the related work in Section~\ref{sec:related} and present our mathematical formulation in Section~\ref{sec:model}.  Our main results are presented in Section~\ref{sec:alg}. Numerical evaluation of our algorithms against benchmarks are given in Section~{\ref{sec:numerical}} and finally we conclude the paper in Section~\ref{sec:con}.

\section{Related Work}\label{sec:related}

Optimizing energy consumption has been widely studied in the literature with a primary focus on clustering \cite{ye05}, routing \cite{manjeshwar01} and MAC protocols \cite{heinzelman00}.  With the proliferation of smart sensors \cite{nazemi16},  in-network data processing, such as data aggregation, has been widely used as a mean of reducing system energy cost by lowering the data volume for transmission.  {Yu et al.} \cite{yu08} proposed an efficient algorithm for data compression in a tree structured networks. {Nazemi et al. } \cite{nazemi16} further presented a distributed algorithm to obtain the optimal compression ratio at each node in a tree structured network so as to minimize the overall energy consumption.

However, none of these works considered caching costs.  As caches have been widely deployed in many modern data communication networks, they can be used as a mean to enhance the system performance by making data available close to end users, which in turn reduces the communication costs \cite{choi12} and latency.

A number of authors have studied the optimization issues for caching allocation \cite{jiansri17,shanmugam13,applegate16,baev08,borst10,ioannidis16,nitishjian17,nitishjianfaheem17,nitishjiantech18}. Ioannidis, Li and Shanmugam et. al \cite{ioannidis16,shanmugam13,jian18} showed that it is NP-hard to determine the optimal data caching location, and an ($1-1/e$) approximation algorithm can be obtained through the pipage rounding algorithm \cite{ageev04,calinescu07}.  Beyond cache placement, \cite{dehghan15} and \cite{ioannidis17} have jointly optimized routing and caching under a single hop bipartite graph and general graph, respectively.  However, none of the existing work considered data compression and the corresponding costs for caching and compression.

{Among all these work, the recent paper by Zafari et al. \cite{faheemjian17} is closer to the problem we tackle here.  {The differences between our work and \cite{faheemjian17} are mainly from two perspectives.}   First,  the mathematical formulations (objectives) are quite different.  Zafari et al. \cite{faheemjian17} considered the energy tradeoffs among communication, compression and caching in communication network, while we focus on maximizing the overall compression and caching gain by characterizing the tradeoff between compression and caching costs with an overall energy consumption constraint.  This difference requires different techniques to handle the problem.   Second, the methodologies are different.  \cite{faheemjian17} aimed to provide a solution to 
the non-convex mixed integer programming problem (MINLP) with an $\epsilon$-global\footnote{{The value of $\epsilon$ depends on the requirement of different problems. Usually it is very small such as $0.0001$, i.e., the obtained solution and global optimal one differ by $0.0001$.}} optimality guarantee.  Since MINLP is NP-hard in general, the proposed algorithm V-SBB in \cite{faheemjian17} is complex and slow to converge to an $\epsilon$-global optimal solution. Furthermore, it is difficult to be generalized to larger network topologies as the algorithm relies on symbolically reformulating the original non-convex MINLP problem that results in extra constraints and variables.  % \red{Their proposed algorithm provides an $\epsilon$-global\footnote{\red{Usually very small in value such as 0.0001, which means that the obtained solution and global optimal differ by 0.0001}} optimality guarantee. However, it cannot be applied to larger network topologies as the algorithm relies on symbolically reformulating the original non-convex MINLP problem that results in extra constraints and variables. Furthermore, the algorithm requires significant amount of time to converge.} 
 Instead, in this paper, we focus on developing an approximation algorithm to optimizing the gain defined above.  In doing so, we first allow the caching decision variables to be continuous, approximate the objective function and then convert the {problem} into a convex one.  Finally, we propose a master-slave based algorithm to efficiently solve the approximated {relaxed} problem, and show that the rounded solutions are feasible to the original problem with performance guarantee.  Compared to the algorithm in \cite{faheemjian17}, our algorithm is more lightweight and efficient and can be applied to a larger problem size. }

%One of the important contributions of this paper is to develop an approximation algorithm that maximizes the overall compression and caching gain by characterizing the tradeoff between compression and caching costs with an overall energy consumption constraint.  
Note that we focus on minimizing the latency and ignore throughput issues, since we do not model congestion.  Combing these two issues together and propos{ing} efficient approximation algorithms is an  interesting problem, which is out of the scope of this paper.

\section{Model}\label{sec:model}

We represent the network as a directed graph $G=(V, E).$ For simplicity, we consider a tree, with $N=|V|$ nodes, as shown in Figure~\ref{fig:1}. Node $v \in V$ is capable of storing $S_v$ amount of data.  Let $\mathcal{K}\subseteq V$ with $K=|\mathcal{K}|$ be the set of leaf nodes, i.e., $\mathcal{K}=\{1, 2, \cdots, K\}$.  Time is partitioned into periods of equal length $T>0$ and data generated in each period are independent.  Without loss of generality (w.l.o.g.), we consider one particular period in the remainder of the paper.}  We assume that only leaf nodes $k\in\mathcal{K}$ can generate data, and all other nodes in the tree receive and compress data from their children nodes, and transmit and/or cache the compressed data to their parent nodes during time $T$.  In Section~\ref{sec:assumption-relax}, we discuss how these assumptions can be relaxed.  For ease of exposition,  the parameters used throughout this paper are summarized in Table~\ref{tab:notations}.

\begin{figure}
\centering
\includegraphics[width=0.45\textwidth]{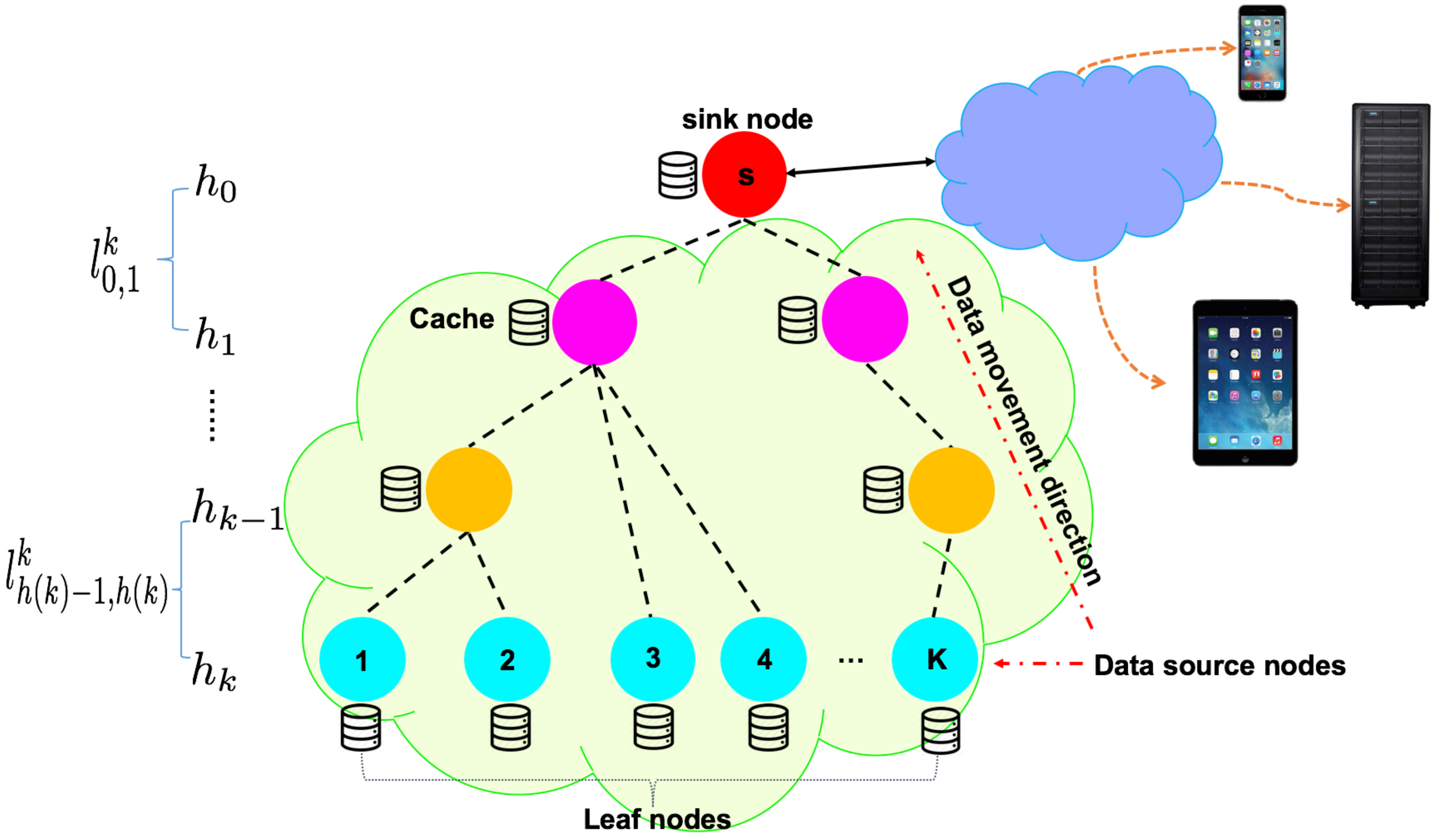}
\caption{Tree-Structured Network Model.}
\label{fig:1}
\vspace{-0.2in}
\end{figure}

\begin{table}
\caption{Summary of Notations}
\vspace{-0.1in}
\begin{tabular}{l|p{6cm}}
\hline
\textbf{Notation} & \textbf{Description} \\ \hline
 $G(V, E)$ & Network graph with $|V|=N$ nodes\\
  $\mathcal{K}$	& Set of leaf nodes with $|\mathcal{K}|=K$  \\ 
  $S_v$ & Cache capacity at node $v\in V$\\
  $h_i^k$ & The $i$-th node on the path between leaf node $k$ and sink node\\ 
  $\delta_{k,i}$ & Compression ratio for data generated by leaf node $k$ at node $i\in V$\\
  $l_{ij}$ & Latency of edge $(i, j)\in E$\\
  $\varepsilon_{vR}$ & per-bit reception cost of node $v\in V$   \\ 
  $\varepsilon_{vT}$ & per-bit transmission cost of node $v\in V$   \\ 
  $\varepsilon_{vC}$& per-bit compression cost of node $v\in V$ \\ 
  $y_{k}$ 	&  Number of data (bits) generated at node $k\in\mathcal{K}$ \\ 
  $b_{k, v}$	&Variable indicating whether node $v\in V$ caches the data from leaf node $k\in\mathcal{K}$ \\ 
   $w_{ca}$ & Caching power efficiency  \\  
   $R_k$ & Request rate for data from node $k\in\mathcal{K}$\\ 
    $W$	& Energy constraint\\ 
   $T$	& Time duration that data are cached\\ 
   {$\delta_{v}$}	& {Reduction rate at node $v$}\\ %\hline
   {$C_{v}$}	& {Set of leaf nodes that are children of node $v$}\\
    {\text{s.t.}}	& {Subject to}\\ \hline
\end{tabular}
\label{tab:notations}
\vspace{-0.15in}
\end{table}

Our objective is to determine the optimal data compression ratio and caching locations across the network to minimize average latency under an energy constraint. 

\subsection{Compression and Caching Costs}
 Let $y_k$ be the amount of data generated by leaf node $k\in\mathcal{K}$. The data generated at the leaf nodes are transmitted up the tree to the sink node $s,$ which serves requests for the data generated in the network.  Let $h(k)$ be the depth of node $k$ in the tree. W.l.o.g., we assume that the sink node is located at level $h(s)=0.$  We represent a path from node $k$  to the sink node as the unique path $\mathcal{H}^k$ of length $h(k)$ as a sequence {$\{h_0^k, h_1^k, \cdots, h_{h(k)}^k\}$} of nodes $h_j^k\in V$ such that $(h_j^k, h_{j+1}^k)\in E,$ where $h_0^k\triangleq s$ (i.e., the sink node) and $h_{h(k)}^k\triangleq k$ (i.e., the node itself).

We denote the per-bit reception, transmission and compression cost of node $v \in V$ as $\varepsilon_{vR}, \varepsilon_{vT}$, and $\varepsilon_{vC},$ respectively.  Each node $h_i^k$ along the path $\mathcal{H}^k$ compresses the data generated by leaf node $k$ at a \emph{data reduction rate\footnote{defined as the ratio of the volume of the output data to the volume of input data at any node}} $\delta_{k,i}$, where $0<\delta_{k, i}\leq 1,$ $\forall i, k.$  
The higher the value of $\delta_{k, i}$, the lower the compression will be, and vice versa. The higher the degree of data compression, the larger will be the amount of energy consumed by compression (computation).  Similarly, caching  data  closer to the sink node can reduce the transmission cost for serving the request, however, each node only has a finite storage capacity. We study the tradeoff among the energy consumed at each node for transmitting, compressing and caching  data to minimize the average delay (which will be defined in \eqref{eq:exact-latency}) in serving a request.  

We consider an energy-proportional model \cite{choi12} for caching, i.e., $w_{ca}\delta_vy_{v}T$ units of energy is consumed if the received data $y_v$ is cached for a duration of $T$  where $w_{ca}$ represents the power efficiency of caching,  which strongly depends on the storage hardware technology. $w_{ca}$ is assumed to be identical for all the nodes.

Data produced by every leaf node $k$ is received, transmitted, and possibly compressed by all nodes in the path from  the leaf  node $k$ to the root node. % the amount of energy consumption is
 On the first request, the energy consumed for this processing of the data from leaf node $k$ is 
\begin{align}
E^{\text{C}}_k= \sum_{i=0}^{h(k)}y_kf(\delta_{k, i})\prod_{m=i+1}^{h(k)}\delta_{k,m},
\label{eq:servingcost}
\end{align}
where $\prod_{m=i}^{j} \delta_{k,m} := 1$ if $i \ge j$ %Equation~\eqref{eq:servingcost} captures one-time\footnote{During every time period  $T$,  data is always pushed towards the sink upon the first request.} energy cost of receiving, compressing and transmitting data $y_k$ from leaf node (level $h(k)$) to the sink node (level $0$).  The amount of data received by any node at level $i$ from leaf node $k$ is $y_k\prod_{m=i+1}^{h(k)}\delta_{k,m}$ due to the compression from level $h(k)$ to $i+1.$  The term $f(\delta_{k,i})$ captures the reception, transmission and compression energy cost for node at level $i$ along the path from leaf node $k$ to the sink node.  
 and $f(\delta_v)= \varepsilon_{vR}+\varepsilon_{vT}\delta_{v}+\varepsilon_{vC}l_v(\delta_{v})$ is the sum of per-bit reception, transmission and compression cost at node $v$ per unit time. {We take $l_v(\delta_{v})=1/\delta_{v}-1$ which was used in \cite{nazemi16,faheemjian17} to capture the  compression cost. }

Let $E_k^{\text{R}}$ be the total energy consumed in responding to the subsequent $(R_k-1)$ requests for the data originally generated by leaf node k.  We have
 \begin{align}
E^{\text{R}}_k&= \sum_{i=0}^{h(k)}y_{k}(R_k-1)\Bigg\{f(\delta_{k, i})\prod_{m=i+1}^{h(k)}\delta_{k,m}\bigg(1-\sum_{j=0}^{i-1}b_{k,j}\bigg)\nonumber\displaybreak[0]\\
&\qquad\qquad\qquad+ \bigg(\prod_{m=i}^{h(k)}\delta_{k,m}\bigg)b_{k,i}{\left(\frac{w_{ca}T}{R_k-1}+\varepsilon_{kT}\right)}\Bigg \},
\label{eq:cachretrievecost}
\end{align}
where $b_{k, j}=1$ if node $j$ caches data generated by $k,$ otherwise $b_{k, j}=0.$  The first term captures the energy cost for reception, transmission and compression up the tree from node $v_{k,i-1}$ to $v_{k, 0}$ and the second term captures the energy cost for storage and transmission by node $v_{k,i}$.  A detailed explanation of equations~(\ref{eq:servingcost}) and~(\ref{eq:cachretrievecost}) is provided in Appendix~\ref{sec:app}.

To consider data generated by all leaf nodes, the total energy consumed in the network is 

 \begin{align}\label{eq:totalenergy}
 &E^{\text{total}}(\boldsymbol{\delta},\boldsymbol{b})\triangleq {\sum_{k \in \mathcal{K}}}\bigg(E^{\text{C}}_k+E^{\text{R}}_k\bigg)\nonumber\\
 &=\sum_{k \in \mathcal{K}}\sum_{i=0}^{h(k)}y_kR_kf(\delta_{k, i})\prod_{m=i+1}^{h(k)}\delta_{k,m}-\sum_{k \in \mathcal{K}}\sum_{i=0}^{h(k)}y_{k}(R_k-1)\nonumber\displaybreak[0]\\
 &\qquad\qquad\cdot f(\delta_{k, i})\prod_{m=i+1}^{h(k)}\delta_{k,m}\sum_{j=0}^{i-1}b_{k,j}+\sum_{k \in \mathcal{K}}\sum_{i=0}^{h(k)}y_{k}(R_k-1)\nonumber\displaybreak[1]\\
 &\qquad\qquad\qquad\qquad\qquad\qquad\cdot\bigg(\prod_{m=i}^{h(k)}\delta_{k,m}\bigg)b_{k,i}{\left(\frac{w_{ca}T}{R_k-1}+\varepsilon_{kT}\right)}\nonumber\displaybreak[2]\\
 &=\sum_{k\in\mathcal{K}}\sum_{i=0}^{h(k)}y_k \Bigg\{R_kf(\delta_{k, i})\prod_{m=i+1}^{h(k)}\delta_{k,m}+\bigg(\prod_{m=i}^{h(k)}\delta_{k,m}\bigg)b_{k, i}(w_{ca}T+\nonumber\displaybreak[3]\\
 &(R_k-1)\varepsilon_{kT})\Bigg\}-\sum_{k \in \mathcal{K}}\sum_{i=0}^{h(k)}y_{k}(R_k-1)f(\delta_{k, i})\prod_{m=i+1}^{h(k)}\delta_{k,m}\sum_{j=0}^{i-1}b_{k,j}\nonumber\\
&\leq\sum_{k\in\mathcal{K}}\sum_{i=0}^{h(k)}y_k \Bigg\{R_kf(\delta_{k, i})\prod_{m=i+1}^{h(k)}\delta_{k,m}+\bigg(\prod_{m=i}^{h(k)}\delta_{k,m}\bigg)b_{k, i}(w_{ca}T+\nonumber\\
&(R_k-1)\varepsilon_{kT})\Bigg\}\triangleq \tilde{E}^{\text{total}}(\boldsymbol{\delta},\boldsymbol{b}),
\end{align}

where $\boldsymbol{\delta}=\{\delta_{k, i}, \forall k \in \mathcal{K}, i=0,\cdots, h(k)\}$ and $\boldsymbol{b}=\{b_{k,i}, \forall k \in \mathcal{K}, i=0,\cdots, h(k)\}$. 

Note that $\tilde{E}^{\text{total}}(\boldsymbol{\delta},\boldsymbol{b})$ is an upper bound of $E^{\text{total}}(\boldsymbol{\delta},\boldsymbol{b})$, which are identical when there is no caching in the network. In the following optimization, we use $\tilde{E}^{\text{total}}(\boldsymbol{\delta},\boldsymbol{b})$ for energy constraint.

\subsection{Latency Performance}
W.l.o.g., we consider the path $\{h_0^k, h_1^k, \cdots, h_{h(k)}^k\}$.  A request for data generated by leaf node $k$ is forwarded along this path from the root node s until it  reaches the node that caches the requested data. Upon finding the requested data, it is propagated along the reverse direction of the path, i.e., carrying the requested data to the sink node where the request originated. To capture the average latency due to data transfer at any particular link, we associate each link with a cost $l_{i, j}$ for $(i, j)\in E$, representing the latency of transmitting the data across the link $(i,j)$.  Denote the latency associated with path $\{h_0^k, h_1^k, \cdots, h_{h(k)}^k\}$ as $\{l_{0,1}^k, l_{1,2}^k,\cdots, l_{h(k)-1, h(k)}^k\}$.

Then the overall latency for all the paths is 
\begin{align}\label{eq:exact-latency}
L(\boldsymbol \delta, \boldsymbol b)=\sum_{k\in\mathcal{K}}\sum_{i=0}^{h(k)-1}\prod_{m=i+1}^{h(k)}\delta_{k,m}y_kR_k l_{i, i+1}^k\prod_{j=0}^i(1-b_{k,j}).
\end{align}
%where $b_{k, j}=1$ if node $j$ caches data generated by $k,$ otherwise $b_{k, j}=0$,  and $\delta_{k, j}\in(0, 1]$ for $\forall k\in\mathcal{K}$ and $j=0, \cdots, h(k).$

\subsection{Optimization}
Our objective is to determine the optimal compression ratio $\boldsymbol{\delta}=\{\delta_{k, i}, \forall k \in \mathcal{K}, i=0,\cdots, h(k)\}$ and data caching location $\boldsymbol{b}=\{b_{k,i}, \forall k \in \mathcal{K}, i=0,\cdots, h(k)\}$ to minimize the expected total latency subject to the energy constraint. That is, 
\begin{subequations}\label{opt:min-latency}
\begin{align}
\min\quad& L(\boldsymbol \delta, \boldsymbol b)\label{eq:obj}\\
\text{s.t.}\quad &\sum_{k\in\mathcal{K}}\sum_{i=0}^{h(k)}y_k \Bigg\{R_kf(\delta_{k, i})\prod_{m=i+1}^{h(k)}\delta_{k,m}+\bigg(\prod_{m=i}^{h(k)}\delta_{k,m}\bigg)b_{k, i}\nonumber \displaybreak[0]\\
&\qquad\qquad\qquad\qquad\cdot(w_{ca}T+(R_k-1)\varepsilon_{kT})\Bigg\}\leq W, \displaybreak[1]\label{eq:constraint1}\\
&  b_{k,i}\in\{0,1\}, \forall  k\in\mathcal{K}, i=0, \cdots, h(k),  \displaybreak[2]\label{eq:constraint2}\\
&\sum_{k \in C_v} b_{k, {h(v)}} y_{k} \prod_{j=h(k)}^{h({v})}\delta_{k,j}\leq S_v, \forall \; v \in V,\displaybreak[3]\label{eq:constraint3}\\
& \sum_{i=0}^{h(k)}b_{k,i} \leq 1,\forall k\in\mathcal{K}.\label{eq:constraint4}
\end{align}
\end{subequations}

Now suppose that there is no compression or caching, then all the requests need to be served from leaf nodes.  The corresponding total latency $L^u$ is given as 
\begin{align}
L^u=\sum_{k\in\mathcal{K}}\sum_{i=0}^{h(k)-1}y_k l_{i, i+1}^kR_k.
\end{align}
Clearly, $L^u$ is an upper bound on the expected total latency. 

Then the \emph{compression and caching gain} is 
\begin{align}\label{gain}
G(\boldsymbol \delta, \boldsymbol b)&=L^u-L(\boldsymbol \delta, \boldsymbol b)\nonumber\\
&=\sum_{k\in\mathcal{K}}\sum_{i=0}^{h(k)-1}R_ky_k l_{i, i+1}^k\Bigg(1-\prod_{m=i+1}^{h(k)}\delta_{k,m}\prod_{j=0}^i(1-b_{k,j})\Bigg).
\end{align}

An \emph{optimization {problem} equivalent} to~(\ref{opt:min-latency}) is to maximize the above gain, given as follows
%\begin{equation}\label{opt:max-gain}
\begin{align}
\max\quad& G(\boldsymbol \delta, \boldsymbol b)\nonumber \\%\label{eq:obj-max} \nonumber \\
\text{s.t.}\quad & \text{Constraints in~(\ref{opt:min-latency}).} 
\label{opt:max-gain}
\end{align}
%\end{equation}
Objective in~(\ref{opt:max-gain}) maximizes the expected compression and caching gain. Constraint~(\ref{eq:constraint1}) ensures that the total energy consumption in the network as given in~(\ref{eq:totalenergy}) is limited.  Constraint~(\ref{eq:constraint2}) indicates that our caching decision variables are binary. Constraint~(\ref{eq:constraint3}) ensures that each cache $v$ stores no more than $S_v$ amount of data. Constraint~(\ref{eq:constraint4}) ensures that at most one copy of the generated data can be cached at any node along the path between the leaf  and the sink node.  Each node potentially {compresses} data from different leaf nodes differently; the coupling occurs due to the storage and energy constraints.

\subsection{Complexity Analysis}
There are two decision variables in~(\ref{opt:max-gain}), i.e., the compression ratio and the caching decision variables. In the following, we will show the impact of these variables on the hardness of our problem, i.e., we consider two cases, (i) given the caching decisions variables $\boldsymbol b;$ (ii) given the compression ratio $\boldsymbol \delta.$

\subsubsection{Given Caching Decisions:}
For given caching decision variables $\boldsymbol b,$ the optimization problem in~(\ref{opt:max-gain}) turns into a geometric programming problem over the compression ratio $\boldsymbol \delta$ that can be solved in polynomial time. %\red{Hence} we obtain a geometric programming problem, which is polynomial-time solvable. 
\begin{theorem}\label{thm:opt-compression}
Given  fixed caching decisions $\boldsymbol b,$ the optimization problem in~(\ref{opt:max-gain}) is polynomial-time solvable. 
\end{theorem}
\begin{proof}
Once $\boldsymbol b$ is given, ~(\ref{opt:max-gain}) becomes a geometric programming problem in $\boldsymbol \delta$, we will show in Section~\ref{sec:convexity} that it can be transformed into a convex optimization problem, which can be solved in polynomial time.
\end{proof}

\subsubsection{Given Compression Ratios:}
Given compression ratios $\boldsymbol \delta$, the optimization problem in~(\ref{opt:max-gain}) is only over the caching decision variables $\boldsymbol b.$ Hence, we obtain an integer programming problem, which is NP-hard.
\begin{theorem}\label{thm:opt-caching}
Given a fixed compression ratio $\boldsymbol \delta,$ the optimization problem in~(\ref{opt:max-gain}) is NP-hard.
\end{theorem}
\begin{proof}
We prove the hardness by  reduction from the classical job-shop problem which is NP-hard \cite{jain99}.

We can reduce the job-shop problem to our problem in~(\ref{opt:max-gain}) with  fixed compression ratios $\boldsymbol \delta$ as follows.  Consider each node $v \in V$ in our model to be a machine $M_i$.  Denote the set of machines as $\mathcal{M}=\{M_1, M_2 \cdots M_{|V|}\}$. The caching decision constitutes the set of jobs $\mathcal{J}=\{J_1, J_2\}$, where $J_1$ means that the data is cached and $J_2$ means otherwise.  Let $\mathcal{X}$ be the set of all sequential job assignments to different machines so that every machine performs every job only once.  The elements $x \in \mathcal{X}$ can be written as $2\times |V|$ matrices, where column $v$ order-wise lists the sequential jobs that the machine $M_v$ will perform. There is a cost function $C$ that captures the cost (i.e., latency) for any machine to perform a particular job.  Our objective in the optimization problem~(\ref{opt:max-gain}) is to find assignments of job $x\in \mathcal{X}$ to minimize the latency or maximize the gain, which is equivalent to the classical job-shop problem.  Since job-shop problem is NP-hard \cite{jain99}, our problem in~(\ref{opt:max-gain}) with given compression ratios $\boldsymbol \delta$  is also NP-hard.
\end{proof}

Therefore, given the results in Theorems~\ref{thm:opt-compression} and~\ref{thm:opt-caching}, we know that our optimization problem is NP-hard in general.
\begin{corollary}
The optimization problem defined in~(\ref{opt:max-gain}) is NP-hard.
\end{corollary}

\subsection{Relaxation of Assumptions}\label{sec:assumption-relax}
We  made several assumptions in the above for the sake of model simplicity. In the following, let us discuss how these assumptions can be relaxed.

First, the network is assumed to be structured as a tree, however, we can easily relax this assumption {by incorporating the route into our optimization problem}. %as long as there is a single unique path from each leaf node to the sink node. 
We take the tree structure as our motivating example since it is a simple and representative topology that captures the key parameters in the optimization formulation without introducing more complexity for a general network topology.

%Second, in our numerical study in Section~\ref{sec:numerical}, all parameters are assumed to be identical across the nodes for the purpose of simplicity. However, our formulation works for node-dependent parameters as well. \red{\bf[comment: We now have results for these. Don't you think we should remove this?]}% as captured by our cost function.%which is not necessary as seen from the cost function. 

Second, while we only allow leaf nodes to generate data, our model can be extended to allow intermediate nodes to generate data at the cost of added complexity, {i.e., the number of decision variables will be increased to represent the caching decision and compression ratio for the data produced at the intermediate nodes}. Furthermore, rather than having a constant $R_k$ requests for data generated at the leaf node $k$, we can generalize our approach to the case where $R_k$ for various leaf nodes are drawn from a distribution such as the Zipf distribution \cite{choi12}. %\red{\bf [comment: This holds true only for our simulation. In our optimization formulation we do not assume $R_k$ to be constant  ]}. 

Third, in our model, we assume that the requests for the data that are generated and valid for a time period $T$ are known.  But our solutions can be applied to an online setting with predicted user requests.

\section{Approximation Algorithm}\label{sec:alg}
Since our optimization problem~(\ref{opt:max-gain}) is NP-hard,  we focus on developing efficient approximation algorithms.  In particular, we develop a polynomial-time solvable algorithm that produces compression ratios and cache decisions with a constant approximation of the minimum average latency.  In the following, we first derive several properties that allow us to develop such an approximation algorithm. Then we discuss how to obtain a constant approximation solution in polynomial time.

\subsection{Properties of the Problem Formulation}
In this section, we show that~(\ref{opt:max-gain}) is a submodular maximization problem under matroid constraints. To begin, we first review the concepts of submodular function and matroid.

\begin{definition}\label{def:submodular}
({Submodular function \cite{schrijver03}}) If $\Omega$ is a finite set, a submodular function is a set function $f$: $2^{\Omega}\rightarrow\mathbb{R}$, where $2^{\Omega}$ denotes the power set of $\Omega$, which satisfies one of the following equivalent conditions:
\begin{enumerate}
\item For every $X, Y\subseteq\Omega$ with $X\subseteq Y$ and every $x\in\Omega\setminus Y$, we have $f(X\cup\{x\})-f(X)\geq f(Y\cup\{x\})-f(Y);$
\item For every $S, T\subseteq\Omega,$ we have $f(S)+f(T)\geq f(S\cup T)+f(S\cap T);$
\item For every $X\subseteq\Omega$ and $x, y\in\Omega\setminus X,$ we have $f(X\cup\{x\})+ f(X\cup\{y\})\geq f(X\cup\{x, y\})+f(X).$
\end{enumerate}
\end{definition}

\begin{definition}\label{def:monotone}
({Monotone sub-modular function \cite{krause2014submodular}}) A sub-modular function $f$ is monotone if for every $T\subseteq S,$ we have $f(T)\geq f(S).$
\end{definition}

\begin{definition}\label{def:matroid}
({Matroid \cite{welsh10}}) A finite matroid $M$ is a pair $(E, \mathcal{I})$, where $E$ is a finite set and $\mathcal{I}$ is a family of subsets of $E$ (called the independent sets) with the following properties:
\begin{enumerate}
\item The empty set is independent, i.e., $\emptyset\in\mathcal{I};$
\item Every subset of an independent set is independent, i.e., for each $A\subset B\subset E,$ if $B\in\mathcal{I}$ then $A\in\mathcal{I};$
\item If $A$ and $B$ are two independent set of $\mathcal{I}$ and $A$ has more elements than $B$, then there exists $x\in A\setminus B$ such that $B\cup\{x\}$ is in $\mathcal{I}$.
\end{enumerate}
\end{definition}

Given the above concepts, we can easily achieve the following result
\begin{theorem}
The objective function in~(\ref{opt:max-gain}) is monotone and sub-modular, and the constraints in~(\ref{opt:max-gain}) are matroid. 
\end{theorem}
The proof is simply to verify that the objective function and constraints in~(\ref{opt:max-gain})  satisfy Definitions~\ref{def:submodular}, ~\ref{def:monotone} and~\ref{def:matroid}. We skip the details due to space limitations.

\begin{corollary}
Since~(\ref{opt:max-gain}) is a sub-modular maximization problem under matroid constraints, a solution with $1/2$ approximation from the optimum can be constructed by a greedy algorithm\footnote{Start with caching all data at the leaf nodes, then compute the optimal compression ratio, and then iteratively add the data to caches by selecting feasible caching decisions at each step that leads to the largest increase in the compression and caching gain.}.
\end{corollary}

%\begin{algorithm}
%\begin{algorithmic}[]
%	\State \textbf{Input}: 
%	\State \textbf{Output}: 
%	\State \red{Starting from that all the data are cached at the leaf nodes, then compute the optimal compression ratio, then iteratively add the data to caches by selecting a feasible cache allocation at each step that leads to the largest increase in the compression and caching gain.}
%	\end{algorithmic}
%	\caption{Greedy Algorithm}
%	\label{algo:greedy}
%\end{algorithm}

%\subsection{Convex Relaxing}\label{sec:convexity}

Now we are ready to develop a polynomial-time solvable approximation algorithm with improved approximation ratio when compared to the greedy algorithm.  Since the optimization problem in~(\ref{opt:max-gain}) is a non-convex mixed integer non-linear programing problem (MINLP), we first relax the integer variables and transform it into a convex optimization problem, which can be solved in polynomial time. Then we round the achieved solutions to ones that satisfy the original integer constraints, if there are any fractional solutions. 

\subsection{Convex Relaxation}\label{sec:convexity}

We first relax the integer variables $b_{k,i}\in\{0, 1\}$ to $\tilde{b}_{k, i}\in[0,1]$ for $\forall k\in\mathcal{K}$ and $i=0,\cdots, h(k),$ in~(\ref{eq:exact-latency}),~(\ref{opt:min-latency}),~(\ref{gain}) and~(\ref{opt:max-gain}).  Let $\mu$ be the joint distribution over $\boldsymbol b,$ and let $\mathbb{P}_{\mu}(\cdot)$ and $\mathbb{E}_{\mu}(\cdot)$ be the corresponding probability and expectation with respect to $\mu,$ i.e., 
\begin{align}
\tilde{b}_{k, i}=\mathbb{P}_{\mu}[b_{k,i}=1]=\mathbb{E}_{\mu}[b_{k,i}].
\end{align}

Then the relaxed expected latency and gain are given as
\begin{align}\label{eq:relaxed-latency-gain}
L(\boldsymbol \delta, \tilde{\boldsymbol b})&=\sum_{k\in\mathcal{K}}\sum_{i=0}^{h(k)-1}\prod_{m=i+1}^{h(k)}\delta_{k,m}y_kR_k l_{i, i+1}^k\prod_{j=0}^i(1-\tilde{b}_{k,j}),\nonumber\\
G(\boldsymbol \delta, \tilde{\boldsymbol b})&=L^u-L(\boldsymbol \delta, \tilde{\boldsymbol b})\nonumber\\
&=\sum_{k\in\mathcal{K}}\sum_{i=0}^{h(k)-1}R_ky_k l_{i, i+1}^k\Bigg(1-\prod_{m=i+1}^{h(k)}\delta_{k,m}\prod_{j=0}^i(1-\tilde{b}_{k,j})\Bigg).
\end{align}
Therefore, the relaxed optimization problem is 
\begin{align}\label{opt:max-gain-relaxed}
\max\quad& G(\boldsymbol \delta, \tilde{\boldsymbol b})\nonumber\\
\text{s.t.}\quad &\sum_{k\in\mathcal{K}}\sum_{i=0}^{h(k)}y_k \Bigg\{R_kf(\delta_{k, i})\prod_{m=i+1}^{h(k)}\delta_{k,m}+\bigg(\prod_{m=i}^{h(k)}\delta_{k,m}\bigg)\tilde{b}_{k, i}\nonumber\displaybreak[0]\\
&\qquad\qquad\qquad\qquad\cdot(w_{ca}T+(R_k-1)\varepsilon_{kT})\Bigg\}\leq W,\nonumber\displaybreak[1]\\
&  \tilde{b}_{k,i}\in[0,1], \forall  k\in\mathcal{K}, i=0, \cdots, h(k),  \nonumber\displaybreak[2]\\
&\sum_{k \in C_v} \tilde{b}_{k, {h(v)}} y_{k} \prod_{j=h(k)}^{h({v})}\delta_{k,j}\leq S_v, \forall \; v \in V,\nonumber\\
& \sum_{i=0}^{h(k)}\tilde{b}_{k,i} \leq 1,\forall k\in\mathcal{K}.
\end{align}

\begin{theorem}\label{thm:relation-oropt-relaxed-opt}
Suppose that $(\boldsymbol \delta^*, \boldsymbol b^*)$ and $(\tilde{\boldsymbol \delta}^*, \tilde{\boldsymbol b}^*)$ are the optimal solutions to~(\ref{opt:max-gain}) and~(\ref{opt:max-gain-relaxed}), respectively, then 
\begin{align}
G(\tilde{\boldsymbol \delta}^*, \tilde{\boldsymbol b}^*)\geq G(\boldsymbol \delta^*, \boldsymbol b^*).
\end{align}
\end{theorem}
\begin{proof}
The results hold since~(\ref{opt:max-gain-relaxed}) maximizes the same objective function over a larger domain due to relaxation of integer variables $\boldsymbol b$ and energy constraint in~(\ref{eq:totalenergy}).
\end{proof}

However, ~(\ref{opt:max-gain-relaxed}) is not a convex optimization problem.  Since $e^x\approx 1+x$ for $x\rightarrow 0$ and $\log(1-x)\approx -x$ for {$x\rightarrow0,$} we obtain an approximation for~(\ref{eq:relaxed-latency-gain}). The approximated expected total latency and \emph{approximated compression and caching gain} are given as follows
\begin{align}\label{eq:approx-latency-gain}
&L(\boldsymbol \delta, \tilde{\boldsymbol b})=\sum_{k\in\mathcal{K}}\sum_{i=0}^{h(k)-1}\prod_{m=i+1}^{h(k)}\delta_{k,m}y_kR_k l_{i, i+1}^k\prod_{j=0}^i(1-\tilde{b}_{k,j})\nonumber\\
&=\sum_{k\in\mathcal{K}}\sum_{i=0}^{h(k)-1}\prod_{m=i+1}^{h(k)}\delta_{k,m}y_kR_k l_{i, i+1}^k e^{\sum_{j=0}^i\log(1-\tilde{b}_{k,j})}\nonumber\displaybreak[0]\\
&\stackrel{(a)}{\approx}\sum_{k\in\mathcal{K}}\sum_{i=0}^{h(k)-1}\prod_{m=i+1}^{h(k)}\delta_{k,m}y_kR_k l_{i, i+1}^k\left(1-\min\left\{1, \sum_{j=0}^i\tilde{b}_{k, j}\right\}\right)\nonumber\displaybreak[1]\\
&\triangleq\tilde{L}(\boldsymbol \delta, \tilde{\boldsymbol b}), \nonumber\displaybreak[2]\\
&\tilde{G}(\boldsymbol \delta, \tilde{\boldsymbol b})=L^u-\tilde{L}(\boldsymbol \delta, \tilde{\boldsymbol b})=\sum_{k\in\mathcal{K}}\sum_{i=0}^{h(k)-1}R_ky_k l_{i, i+1}^k\Bigg(1-\prod_{m=i+1}^{h(k)}\delta_{k,m}\nonumber\displaybreak[3]\\
&\qquad\qquad\qquad\qquad\qquad\qquad\quad\cdot\left(1-\min\left\{1, \sum_{j=0}^i\tilde{b}_{k, j}\right\}\right)\Bigg),
\end{align}
where $(a)$ is based on the two approximate properties discussed above.

Then, the \emph{relaxed approximated optimization problem} is given as 
\begin{subequations}
\begin{align}
\max\quad& \tilde{G}(\boldsymbol \delta,  \tilde{\boldsymbol b})\label{opt:approx-max-gain-relaxed}\\
\text{s.t.}\quad &\sum_{k\in\mathcal{K}}\sum_{i=0}^{h(k)}y_k \Bigg\{R_kf(\delta_{k, i})\prod_{m=i+1}^{h(k)}\delta_{k,m}+\bigg(\prod_{m=i}^{h(k)}\delta_{k,m}\bigg)\tilde{b}_{k, i}\nonumber\displaybreak[0]\\
&\qquad\qquad\qquad\qquad\cdot(w_{ca}T+(R_k-1)\varepsilon_{kT})\Bigg\}\leq W,\displaybreak[1]\label{eq:constraint-appr1}\\
& \tilde{b}_{k,i}\in[0,1], \forall  k\in\mathcal{K}, i=0, \cdots, h(k),  \displaybreak[2]\label{eq:constraint-appr2}\\
&\sum_{k \in C_v} \tilde{b}_{k, {h(v)}} y_{k} \prod_{j=h(k)}^{h({v})}\delta_{k,j}\leq S_v, \forall \; v \in V, \displaybreak[3]\label{eq:constraint-appr3}\\
& \sum_{i=0}^{h(k)}\tilde{b}_{k,i} \leq 1,\forall k\in\mathcal{K}.\label{eq:constraint-appr4}
\end{align}
\end{subequations}

However, $\tilde{G}(\boldsymbol \delta, \tilde{\boldsymbol b})$ is not concave. In the following, we transform it into a convex term through Boyd's method (Section $4.5$ \cite{boyd04}) to deal with posynomial terms in~(\ref{opt:approx-max-gain-relaxed}), ~(\ref{eq:constraint-appr1}) and~(\ref{eq:constraint-appr3}).

\subsubsection{Transformation of the Objective Function}
Given our approximated objective function 
 \begin{align}
\tilde{L}(\boldsymbol \delta, \tilde{\boldsymbol b}) \triangleq\sum_{k\in\mathcal{K}}\sum_{i=0}^{h(k)-1}\prod_{m=i+1}^{h(k)}\delta_{k,m}y_kR_k l_{i, i+1}^k\left(1-\min\left\{1, \sum_{j=0}^i\tilde{b}_{k, j}\right\}\right),
\end{align}
we define two new variables as follows
\begin{align}\label{newvariables}
\log(\tilde{b}_{k,j})\triangleq u_{k,j},&\quad i.e.,\quad \tilde{b}_{k,j}=e^{u_{k,j}},\nonumber\\
\log\delta_{k,m}\triangleq\tau_{k,m},&\quad i.e.,\quad \delta_{k,m}=e^{\tau_{k,m}}.
\end{align}

Then the approximated objective function can be transformed into
\begin{small}
 \begin{align}
\tilde{L}(\boldsymbol \tau, \boldsymbol u) \triangleq\sum_{k\in\mathcal{K}}\sum_{i=0}^{h(k)-1}\sum_{m=i+1}^{h(k)}e^{\tau_{k,m}+\log( y_kR_k l_{i, i+1}^k)}\left(1-\min\left\{1, \sum_{j=0}^ie^{u_{k,j}}\right\}\right).
\end{align}
\end{small}

Therefore, we can transform $\tilde{G}(\boldsymbol \delta, \tilde{\boldsymbol b})$ into 
\begin{small}
\begin{align}\label{eq:transfomed-gain}
&\tilde{G}(\boldsymbol \tau, \boldsymbol u)=L^u-\tilde{L}(\boldsymbol \tau, \boldsymbol u)\nonumber\\
&=\sum_{k\in\mathcal{K}}\sum_{i=0}^{h(k)-1}e^{\log (R_ky_k l_{i, i+1}^k)}\left(1-\sum_{m=i+1}^{h(k)}e^{\tau_{k,m}}\left(1-\min\left\{1, \sum_{j=0}^ie^{u_{k, j}}\right\}\right)\right).
\end{align}
\end{small}

Next we need to transform the constraints following Boyd's method.

\subsubsection{Transformation of the Constraints}
${}$\\

\noindent{\textbf{Constraint~(\ref{eq:constraint-appr1}):}} We take the {left} hand side of the constraint and transform it. To simplify, we divide the equation into multiple parts, 
\begin{align}
&\underbrace{\sum_{k\in\mathcal{K}}\sum_{i=0}^{h(k)}R_ky_k f(\delta_{k, i})\prod_{m=i+1}^{h(k)}\delta_{k,m}}_{\text{Part $1$}} +\underbrace{\sum_{k\in\mathcal{K}}\sum_{i=0}^{h(k)}y_k w_{ca}T\tilde{b}_{k, i} \prod_{m=i}^{h(k)}\delta_{k,m}}_{\text{Part $2$}} \nonumber\displaybreak[0]\\
&\qquad\qquad\qquad+ \underbrace{\sum_{k\in\mathcal{K}}\sum_{i=0}^{h(k)}y_k\varepsilon_{kT} (R_k-1)\tilde{b}_{k, i}\prod_{m=i}^{h(k)}\delta_{k,m}  }_{\text{Part $3$}}.
\end{align}

\noindent{{\textbf{Part $1$:}} From~(\ref{newvariables}), i.e., $\tau_{k,i}=\log\delta_{k,i}$,  we have 
\begin{align}\label{eq:con1part1}
&\textbf{Part $1$}= \sum_{k\in\mathcal{K}}\sum_{i=0}^{h(k)}R_ky_k (\varepsilon_{kR}-\varepsilon_{kC}+\delta_{k,i}\varepsilon_{kT}+\frac{\varepsilon_{kC}}{\delta_{k,i}})\prod_{m=i+1}^{h(k)}\delta_{k,m} \nonumber \displaybreak[0]\\
 = &\sum_{k\in\mathcal{K}}\sum_{i=0}^{h(k)}R_ky_k (\varepsilon_{kR}-\varepsilon_{kC}+e^{\tau_{k,i}}\varepsilon_{kT}+\frac{\varepsilon_{kC}}{e^{\tau_{k,i}}})\prod_{m=i+1}^{h(k)}\delta_{k,m} \nonumber \displaybreak[1]\\
=  &\sum_{k\in\mathcal{K}}\sum_{i=0}^{h(k)}R_ky_k(\varepsilon_{kR}-\varepsilon_{kc}+\varepsilon_{kT}e^{\tau_{k,i}}+\varepsilon_{kc}e^{-\tau_{k,i}} ) e^{\sum_{m=i+1}^{h(k)}\tau_{k,m}}. 
\end{align}

\noindent{\textbf{Part $2$:}} From~(\ref{newvariables}), i.e., $\tilde{b}_{k,j}=e^{u_{k,j}}$, we have 
\begin{align}\label{eq:con1part2}
\textbf{Part $2$}=\sum_{k\in\mathcal{K}}\sum_{i=0}^{h(k)}e^{\sum_{m=i}^{h(k)}\tau_{k,m}+\log(y_kw_{ca}T)+u_{k, i}}.
\end{align}

\noindent{\textbf{Part $3$:}}  
Similarly, we have
\begin{align}\label{eq:con1part3}
\textbf{Part $3$}= \sum_{k\in\mathcal{K}}\sum_{i=0}^{h(k)}e^{\sum_{m=i}^{h(k)}\tau_{k,m}+\log \; (y_k(R_{k}-1)\varepsilon_{kT})+u_{k, i}}. 
\end{align}

Combining~(\ref{eq:con1part1}),~(\ref{eq:con1part2}) and~(\ref{eq:con1part3}), Constraint~(\ref{eq:constraint-appr1}) becomes 
\begin{align}\label{convexcon1}
&\sum_{k\in\mathcal{K}}\sum_{i=0}^{h(k)}{R_k}(y_k\varepsilon_{kR}-y_k\varepsilon_{kc}+y_k\varepsilon_{kT}e^{\tau_{k,i}}+y_k\varepsilon_{kc}e^{-\tau_{k,i}} ) e^{\sum_{m=i+1}^{h(k)}\tau_{k,m}} \nonumber\displaybreak[0]\\
&\qquad\qquad+\sum_{k\in\mathcal{K}}\sum_{i=0}^{h(k)}e^{\sum_{m=i}^{h(k)}\tau_{k,m}+\log(y_kw_{ca}T)+u_{k, i}}\nonumber\displaybreak[1]\\
&\qquad\qquad+\sum_{k\in\mathcal{K}}\sum_{i=0}^{h(k)}e^{\sum_{m=i}^{h(k)}\tau_{k,m}+\log \; (y_k(R_{k}-1)\varepsilon_{kT})+u_{k, i}}\leq W,
\end{align}
which is convex in $\boldsymbol \tau$ and $\boldsymbol u$ on the left hand side, respectively.

\noindent{\textbf{Constraint~(\ref{eq:constraint-appr3}):}} Similarly, we have
 \begin{align}\label{eq:con3}
\sum_{k \in C_v} e^{\sum_{j=h(k)}^{h(v)}\tau_{k,j}+\log y_k+u_{k, h(v)}}\leq S_v,
 \end{align}
which is convex in $\boldsymbol \tau$ and $\boldsymbol u$ on the left hand side, respectively.

\subsubsection{Optimization Problem in Convex Form}
Following the transformation given in~(\ref{eq:transfomed-gain}),~(\ref{convexcon1}) and~(\ref{eq:con3}), we obtain the convex form for the optimization problem, i.e., 
\begin{align}\label{opt:convextransformed}
\max\quad&\tilde{G}(\boldsymbol \tau, \boldsymbol u)\nonumber\\
\text{s.t.}\quad& \sum_{k\in\mathcal{K}}\sum_{i=0}^{h(k)}R_ky_k(\varepsilon_{kR}-\varepsilon_{kc}+\varepsilon_{kT}e^{\tau_{k,i}}+\varepsilon_{kc}e^{-\tau_{k,i}} ) e^{\sum_{m=i+1}^{h(k)}\tau_{k,m}} \nonumber\\
&\qquad\qquad+\sum_{k\in\mathcal{K}}\sum_{i=0}^{h(k)}e^{\sum_{m=i}^{h(k)}\tau_{k,m}+\log(y_kw_{ca}T)+u_{k, i}}\nonumber\displaybreak[0]\\
&\qquad\qquad+\sum_{k\in\mathcal{K}}\sum_{i=0}^{h(k)}e^{\sum_{m=i}^{h(k)}\tau_{k,m}+\log \; (y_k(R_{k}-1)\varepsilon_{kT})+u_{k, i}}\leq W, \nonumber \displaybreak[1]\\
&e^{u_{k, i}}\in[0,1], \forall  k\in\mathcal{K}, i=0, \cdots, h(k),  \nonumber\displaybreak[2]\\
 &\sum_{k \in C_v} e^{\sum_{j=h(k)}^{h(v)}\tau_{k,j}+\log y_k+u_{k, h(v)}}\leq S_v,\nonumber\displaybreak[3]\\
 & \sum_{i=0}^{h(k)}e^{u_{k,i}} \leq 1,\forall k\in\mathcal{K}.
\end{align}

\begin{theorem}
The optimization problem given in~(\ref{opt:convextransformed}) is convex in $\boldsymbol \tau$ and $\boldsymbol u$, respectively. 
\end{theorem}
\begin{proof}
It can be easily checked that the objective function in~(\ref{opt:convextransformed}) satisfies the second order condition \cite{boyd04} for $\boldsymbol \tau$ and $\boldsymbol u,$ respectively. We omit the details due to space constraints.  
\end{proof}

\begin{remark}
Note that the optimization problem given in~(\ref{opt:convextransformed}) is convex in $\boldsymbol \tau$ for a given $\boldsymbol u,$ and vice versa.  In the following, we will present an efficient master-slave algorithm to solve the convex optimization problem in $\boldsymbol \tau$ and $\boldsymbol u,$ respectively.
\end{remark}

%\subsection{Efficient Algorithms}\label{sec:prop}

\subsection{Efficient Algorithms}\label{sec:prop}

\begin{theorem}\label{thm:equivalence}
The optimization problems given in~(\ref{opt:approx-max-gain-relaxed}) and~(\ref{opt:convextransformed}) are equivalent. 
\end{theorem}
 \begin{proof}
 This is clear from the way we convexified the problem.
 \end{proof}
 
Note that after the convex relaxation and transformation, the optimization problem in~(\ref{opt:convextransformed}) is point-wise convex in $\boldsymbol \tau$ and $\boldsymbol u.$  We focus on designing a polynomial-time solvable algorithm. 
 
\noindent{\textit{\textbf{Algorithm:}}}  %Since the optimization problem in~(\ref{opt:convextransformed}) is convex in $\boldsymbol \tau$ and $\boldsymbol u,$ respectively.  
 We consider the master-slave algorithm shown in Algorithm~\ref{algo:alg}, i.e., given a fixed $\boldsymbol \tau_0$, we solve~(\ref{opt:convextransformed}) to obtain $\boldsymbol u_0,$ and then given $\boldsymbol u_0$, we solve~(\ref{opt:convextransformed}) to obtain $\boldsymbol \tau_1.$  We repeat the above process until that the values of $\boldsymbol \tau$ and $\boldsymbol u$ converge\footnote{If the difference between the current value and the previous one is within a tolerance, we say the value converges.}\footnote{Since our objective function is a function of the variables $\boldsymbol u$ and $\boldsymbol \tau$, once these variables converge, the value of the objective function must converge. As we are interested in the objective value, in Algorithm~\ref{algo:alg}, we write the convergence criteria with respect to the objective function value, where $\epsilon$ equals to $0.001.$}. We denote this as the optimal solution of~(\ref{opt:convextransformed}) as $(\boldsymbol \tau^*, \boldsymbol u^*)$\footnote{Note that our master-slave algorithm is very efficient to solve this convex optimization problem, we can obtain a solution within one or two iterations.}.

\begin{algorithm}
\begin{algorithmic}[]
	\State \textbf{Input}: $R_k$, $y_k$, $W$, $\boldsymbol l, \text{obj}_0$ 
	\State \textbf{Output}: $\boldsymbol {b, \delta}, \text{obj}_f$
	\State \textbf{Step $1$:} Initialize $\boldsymbol u$
	\State \textbf{Step $2$}: $\boldsymbol \tau$ $\longleftarrow$ $\text{Random}(\boldsymbol{l{b}},\boldsymbol{u{b}})$ \Comment{Generate random $\boldsymbol{\tau}$ between {lower bound} $\boldsymbol{l{b}}$ and {upper bound} $\boldsymbol{u{b}}$}
	\While{$\text{obj}_{\chi}-\text{obj}_{\chi-1}\geq \epsilon$}
	\State \textbf{Step $3$}: $\boldsymbol u_{\chi}\longleftarrow \text{Convex}(\text{master},\boldsymbol \tau_{\chi})$ \Comment{Solve the master optimization problem for $\boldsymbol u_{\chi}$  }
	\State \textbf{Step $4$}: $\boldsymbol \tau_{\chi} \longleftarrow \text{Convex}(\text{slave},\boldsymbol u_{\chi})$ \Comment{Solve the slave optimization problem for $\boldsymbol \tau_{\chi}$  }
	\State \textbf{Step $5$}: $(\boldsymbol b_{\chi}, \boldsymbol{\delta_{\chi}}, \text{obj}_{\chi})  \longleftarrow \text{Rounding}(\boldsymbol u_{\chi},\boldsymbol \tau_{\chi})$ \Comment{Round the values of $\boldsymbol u_{\chi}$, remap $\boldsymbol u_{\chi}, \boldsymbol \tau_{\chi}$ to $\boldsymbol{b}_{\chi}$ and $ \boldsymbol{\delta}_{\chi}$ and obtain the new objective function value  }
	\EndWhile
	\end{algorithmic}
	\caption{Master-Slave Algorithm}
	\label{algo:alg}
\end{algorithm}

Given the optimal solution to~(\ref{opt:convextransformed}) as $(\boldsymbol \tau^*, \boldsymbol u^*)$, then from Theorem~\ref{thm:equivalence}, we know there exists $(\boldsymbol \delta^{**}, \tilde{\boldsymbol b}^{**})$, which is the optimal solution to~(\ref{opt:approx-max-gain-relaxed}) such that $\tilde{G}(\boldsymbol \tau^*, \boldsymbol u^*)=\tilde{G}(\boldsymbol \delta^{**}, \tilde{\boldsymbol b}^{**})$ and $G(\boldsymbol \tau^*, \boldsymbol u^*)=G(\boldsymbol \delta^{**}, \tilde{\boldsymbol b}^{**}).$

\begin{theorem}\label{thm:relation-relaxed-opt-convexed-opt}
Denote the optimal solutions to~(\ref{opt:max-gain-relaxed}) and~(\ref{opt:convextransformed}) as $(\tilde{\boldsymbol \delta}^*, \tilde{\boldsymbol b}^*)$ and $(\boldsymbol \tau^*, \boldsymbol u^*)$, respectively. Then, we have 
\begin{align}
\left(1-\frac{1}{e}\right)G(\tilde{\boldsymbol \delta}^*, \tilde{\boldsymbol b}^*)\leq G(\boldsymbol \tau^*, \boldsymbol u^*)\leq G(\tilde{\boldsymbol \delta}^*, \tilde{\boldsymbol b}^*).
\end{align}
\end{theorem}
\begin{proof}
Consider any $(\boldsymbol \delta, \tilde{\boldsymbol b})$ that satisfies the constraints in~(\ref{opt:max-gain-relaxed}) and~(\ref{opt:approx-max-gain-relaxed}).

First, we show that $G(\boldsymbol \delta, \tilde{\boldsymbol b})\leq \tilde{G}(\boldsymbol \delta, \tilde{\boldsymbol b}),$ as follows 
\begin{align}
G(\boldsymbol \delta, \tilde{\boldsymbol b})\stackrel{(a)}{=}&\sum_{k\in\mathcal{K}}\sum_{i=0}^{h(k)-1}R_ky_k l_{i, i+1}^k\mathbb{E}\Bigg(1-\prod_{m=i+1}^{h(k)}\delta_{k,m}\prod_{j=0}^i(1-b_{k,j})\Bigg)\nonumber\displaybreak[0]\\
=&\sum_{k\in\mathcal{K}}\sum_{i=0}^{h(k)-1}R_ky_k l_{i, i+1}^k-\sum_{k\in\mathcal{K}}\sum_{i=0}^{h(k)-1}R_ky_k l_{i, i+1}^k\nonumber\displaybreak[1]\\
&\qquad\qquad\qquad\cdot\prod_{m=i+1}^{h(k)}\delta_{k,m}\mathbb{E}\left[\prod_{j=0}^i(1-b_{k,j})\right]\nonumber\displaybreak[2]\\
=&\sum_{k\in\mathcal{K}}\sum_{i=0}^{h(k)-1}R_ky_k l_{i, i+1}^k-\sum_{k\in\mathcal{K}}\sum_{i=0}^{h(k)-1}R_ky_k l_{i, i+1}^k\nonumber\displaybreak[3]\\
&\qquad\qquad\qquad\cdot\prod_{m=i+1}^{h(k)}\delta_{k,m}\mathbb{E}\left[1-\min\left\{1, \sum_{j=0}^ib_{k,j}\right\}\right]\nonumber\\
\stackrel{(b)}{\leq}&\sum_{k\in\mathcal{K}}\sum_{i=0}^{h(k)-1}R_ky_k l_{i, i+1}^k-\sum_{k\in\mathcal{K}}\sum_{i=0}^{h(k)-1}R_ky_k l_{i, i+1}^k\nonumber\\
&\qquad\qquad\qquad\cdot\prod_{m=i+1}^{h(k)}\delta_{k,m}\left(1-\min\left\{1, \mathbb{E}\left[\sum_{j=0}^ib_{k,j}\right]\right\}\right)\nonumber\\
=&\tilde{G}(\boldsymbol \delta, \tilde{\boldsymbol b}),
\end{align}
where the expectation $\mathbb{E}$ in (a) is taken over $\boldsymbol b$ due to the linear relaxation, and (b) holds true due to the concavity of the min operator.

Next, we show that  $G(\boldsymbol \delta, \tilde{\boldsymbol b})\geq \left(1-\frac{1}{e}\right)\tilde{G}(\boldsymbol \delta, \tilde{\boldsymbol b}),$  as follows
\begin{align}\label{compare1}
&G(\boldsymbol \delta, \tilde{\boldsymbol b})=\sum_{k\in\mathcal{K}}\sum_{i=0}^{h(k)-1}R_ky_k l_{i, i+1}^k\Bigg(1-\prod_{m=i+1}^{h(k)}\delta_{k,m}\prod_{j=0}^i(1-\tilde{b}_{k,j})\Bigg)\nonumber\displaybreak[0]\\
&\geq \sum_{k\in\mathcal{K}}\sum_{i=0}^{h(k)-1}R_ky_k l_{i, i+1}^k\Bigg(1-\prod_{j=0}^i(1-\tilde{b}_{k,j})\Bigg)\nonumber\displaybreak[1]\\
&\stackrel{(a)}{\geq} \sum_{k\in\mathcal{K}}\sum_{i=0}^{h(k)-1}R_ky_k l_{i, i+1}^k\left(1-(1-1/i)^i\right)\min\left\{1, \sum_{j=0}^i\tilde{b}_{k,j}\right\}\nonumber\displaybreak[2]\\
&\stackrel{(b)}{\geq} \left(1-\frac{1}{e}\right)\sum_{k\in\mathcal{K}}\sum_{i=0}^{h(k)-1}R_ky_k l_{i, i+1}^k\min\left\{1, \sum_{j=0}^i\tilde{b}_{k,j}\right\},
\end{align}
where (a) holds true since  \cite{cornnejols77, goemans94} 
\begin{align}
1-\prod_{j=0}^i(1-\tilde{b}_{k,j})\geq\left(1-(1-1/i)^i\right)\min\left\{1, \sum_{j=0}^i\tilde{b}_{k,j}\right\},
\end{align}
and (b) holds true since $(1-1/i)^i\leq 1/e$. Also we have
\begin{align}\label{compare2}
&\left(1-\frac{1}{e}\right)\tilde{G}(\boldsymbol \delta, \tilde{\boldsymbol b})=\left(1-\frac{1}{e}\right)\sum_{k\in\mathcal{K}}\sum_{i=0}^{h(k)-1}R_ky_k l_{i, i+1}^k-\left(1-\frac{1}{e}\right)\tilde{L}(\boldsymbol \delta, \tilde{\boldsymbol b})\nonumber\displaybreak[0]\\
=&\left(1-\frac{1}{e}\right)\sum_{k\in\mathcal{K}}\sum_{i=0}^{h(k)-1}R_ky_k l_{i, i+1}^k-\left(1-\frac{1}{e}\right)\sum_{k\in\mathcal{K}}\sum_{i=0}^{h(k)-1}\nonumber\displaybreak[1]\\
&\qquad\qquad\qquad\cdot\prod_{m=i+1}^{h(k)}\delta_{k,m}y_kR_k l_{i, i+1}^k\left(1-\min\left\{1, \sum_{j=0}^i\tilde{b}_{k, j}\right\}\right)\nonumber\displaybreak[2]\\
=&\left(1-\frac{1}{e}\right)\sum_{k\in\mathcal{K}}\sum_{i=0}^{h(k)-1}R_ky_k l_{i, i+1}^k-\left(1-\frac{1}{e}\right)\sum_{k\in\mathcal{K}}\sum_{i=0}^{h(k)-1}y_kR_k l_{i, i+1}^k\nonumber\displaybreak[3]\\
&\qquad\qquad\qquad\qquad\cdot\left(1-\min\left\{1, \sum_{j=0}^i\tilde{b}_{k, j}\right\}\right)\nonumber\\
=&\left(1-\frac{1}{e}\right)\sum_{k\in\mathcal{K}}\sum_{i=0}^{h(k)-1}y_kR_k l_{i, i+1}^k\min\left\{1, \sum_{j=0}^i\tilde{b}_{k, j}\right\},
\end{align}
then from~(\ref{compare1}) and~(\ref{compare2}), we immediately have 
\begin{align}
G(\boldsymbol \delta, \tilde{\boldsymbol b})\geq \left(1-\frac{1}{e}\right)\tilde{G}(\boldsymbol \delta, \tilde{\boldsymbol b}),
\end{align}
therefore, for any $(\boldsymbol \delta, \tilde{\boldsymbol b})$ that satisfies the constraints in~(\ref{opt:max-gain-relaxed}) and~(\ref{opt:approx-max-gain-relaxed}), we have
\begin{align}\label{eq:relations-on-gain}
\left(1-\frac{1}{e}\right)\tilde{G}(\boldsymbol \delta, \tilde{\boldsymbol b})\leq G(\boldsymbol \delta, \tilde{\boldsymbol b})\leq \tilde{G}(\boldsymbol \delta, \tilde{\boldsymbol b}).
\end{align}

Now, since $(\tilde{\boldsymbol \delta}^*, \tilde{\boldsymbol b}^*)$ is optimal to~(\ref{opt:max-gain-relaxed}), then
\begin{align}
G(\boldsymbol \delta^{**}, \tilde{\boldsymbol b}^{**})\leq G(\tilde{\boldsymbol \delta}^*, \tilde{\boldsymbol b}^*).
\end{align}
Similarly, since $(\boldsymbol \delta^{**}, \tilde{\boldsymbol b}^{**})$ is optimal to~(\ref{opt:approx-max-gain-relaxed}), 
\begin{align}
G(\tilde{\boldsymbol \delta}^*, \tilde{\boldsymbol b}^*)\leq \tilde{G}(\tilde{\boldsymbol \delta}^*, \tilde{\boldsymbol b}^*)\leq \tilde{G}(\boldsymbol \delta^{**}, \tilde{\boldsymbol b}^{**})\leq \frac{e}{e-1}G(\boldsymbol \delta^{**}, \tilde{\boldsymbol b}^{**}),
\end{align}
where the first and third inequality hold due to~(\ref{eq:relations-on-gain}).

Therefore, we have 
\begin{align}
\left(1-\frac{1}{e}\right)G(\tilde{\boldsymbol \delta}^*, \tilde{\boldsymbol b}^*)\leq G(\boldsymbol \delta^{**}, \tilde{\boldsymbol b}^{**})\leq G(\tilde{\boldsymbol \delta}^*, \tilde{\boldsymbol b}^*),
\end{align}
i.e.,
\begin{align}
\left(1-\frac{1}{e}\right)G(\tilde{\boldsymbol \delta}^*, \tilde{\boldsymbol b}^*)\leq G(\boldsymbol \tau^*, \boldsymbol u^*)\leq G(\tilde{\boldsymbol \delta}^*, \tilde{\boldsymbol b}^*).
\end{align}
\end{proof}

Since~(\ref{opt:convextransformed}) is a convex optimization problem, $(\boldsymbol \tau^*, \boldsymbol u^*)$ can be obtained in strongly polynomial time.

%\subsection{Rounding}\label{sec:rounding}

\subsection{Rounding}\label{sec:rounding}
To provide a constant approximation solution to~(\ref{opt:max-gain}), the optimal solution $(\boldsymbol \delta^{**}, \tilde{\boldsymbol b}^{**})$ needs to be rounded. 

\noindent{\textbf{Property:}} W.l.o.g., we consider a feasible solution $(\boldsymbol \delta, \tilde{\boldsymbol b})$ and assume that there are two fractional solutions $\tilde{b}_{k, j}$ and $\tilde{b}_{k, l}$. We define 
\begin{align}\label{eq:rounding-step}
\epsilon_1=\min\{\tilde{b}_{k, j}, 1-\tilde{b}_{k, l}\},\nonumber\displaybreak[0]\\
\epsilon_2=\min\{1-\tilde{b}_{k, j}, \tilde{b}_{k, l}\},
\end{align}
and set 
\begin{align}\label{eq:rounding-result}
\tilde{\boldsymbol b}^\prime(1)=(\tilde{\boldsymbol b}_{-(j, l)}, \tilde{b}_{k, j}-\epsilon_1, \tilde{b}_{k, l}+\epsilon_1),\nonumber\displaybreak[0]\\
\tilde{\boldsymbol b}^\prime(2)=(\tilde{\boldsymbol b}_{-(j, l)}, \tilde{b}_{k, j}+\epsilon_2, \tilde{b}_{k, l}-\epsilon_2),
\end{align}
where $\tilde{\boldsymbol b}_{-(j, l)}$ means all other components in $\tilde{\boldsymbol b}$ remain the same besides $\tilde{b}_{k, j}$ and $\tilde{b}_{k, l}$.  Set $\tilde{\boldsymbol b}=\tilde{\boldsymbol b}^\prime(1),$ if $G(\tilde{\boldsymbol b}^\prime(1))>G(\tilde{\boldsymbol b}^\prime(2)),$ otherwise set $\tilde{\boldsymbol b}=\tilde{\boldsymbol b}^\prime(2).$ 

\begin{remark}
From the above rounding steps~(\ref{eq:rounding-step}) and~(\ref{eq:rounding-result}), it is clear that $\tilde{\boldsymbol b}^\prime$ has smaller number of fractional components than $\tilde{\boldsymbol b}.$ Since the number of components in $\tilde{\boldsymbol b}$ is finite, the rounding steps will terminate in a finite number of steps.  Also, it is clear that $\tilde{\boldsymbol b}^\prime$ satisfies the second and the fourth constraints in~(\ref{opt:max-gain-relaxed}) and~(\ref{opt:approx-max-gain-relaxed}) for $\forall\epsilon\in[-\epsilon_1, \epsilon_2]$ or $\forall\epsilon\in[-\epsilon_2, \epsilon_1].$
\end{remark}

Now suppose that $(\boldsymbol \delta, \tilde{\boldsymbol b}^\prime)$ is the rounded solution. Then following an argument similar to that in \cite{ageev04}, we have 
\begin{lemma}\label{lem:fea}
For $k\in\mathcal{K}$, if $\sum_{j=1}^{h(k)} b_{k, j}$ is an integer, then $\sum_{j=1}^{h(k)} b_{k, j}^\prime$ is also an integer; if $\sum_{j=1}^{h(k)} b_{k, j}$ is a fraction, then $\left\lfloor\sum_{j=1}^{h(k)} b_{k, j}\right\rfloor\leq \sum_{j=1}^{h(k)} b_{k, j}^\prime\leq\left\lfloor\sum_{j=1}^{h(k)} b_{k, j}\right\rfloor+1.$
\end{lemma}
We refer the interested reader to \cite{ageev04} for more details. 

Now since the energy constraint is integer, given that $(w_{ca}T+(R_k-1)\varepsilon_{kT})\leq 1$, Lemma~(\ref{lem:fea}) implies that $(\boldsymbol \delta, \tilde{\boldsymbol b}^\prime)$ satisfies the constraints in~(\ref{opt:max-gain-relaxed}). Therefore, after the rounding, we obtain a feasible solution obeying the constraints in~(\ref{opt:max-gain-relaxed}).

\begin{theorem}
We consider a feasible solution $(\boldsymbol \delta, \tilde{\boldsymbol b})$ and assume that there are two fractional solutions $\tilde{b}_{k, j}$ and $\tilde{b}_{k, l}$. W.l.o.g., we assume that $\tilde{\boldsymbol b}^\prime=(\tilde{\boldsymbol b}_{-(j, l)}, \tilde{b}_{k, j}-\epsilon, \tilde{b}_{k, l}+\epsilon)$ following rounding steps~(\ref{eq:rounding-step}) and~(\ref{eq:rounding-result}), then $G(\cdot)$ is convex in $\epsilon$.
\end{theorem}
\begin{proof}
Recall that 
\begin{align*}
G(\boldsymbol \delta, \tilde{\boldsymbol b})=\sum_{k\in\mathcal{K}}\sum_{i=0}^{h(k)-1}R_ky_k l_{i, i+1}^k\Bigg(1-\prod_{m=i+1}^{h(k)}\delta_{k,m}\prod_{j=0}^i(1-\tilde{b}_{k,j})\Bigg),
\end{align*}
then
\begin{align*}
G(\boldsymbol \delta, \tilde{\boldsymbol b}^\prime, \epsilon)&=\sum_{k\in\mathcal{K}}\sum_{i=0}^{h(k)-1}R_ky_k l_{i, i+1}^k\Bigg(1-\prod_{m=i+1}^{h(k)}\delta_{k,m}\prod_{j^\prime\neq j, l}^i(1-\tilde{b}_{k,j^\prime})\nonumber\displaybreak[0]\\
&\qquad\qquad\qquad\qquad\qquad\cdot(1-\tilde{b}_{k,j}+\epsilon)(1-\tilde{b}_{k,l}-\epsilon)\Bigg),
\end{align*}
by the second order condition, it is obvious that $G(\cdot)$ is convex in $\epsilon.$ This property is called $\epsilon$-convexity property in \cite{ageev04}.
\end{proof}

\begin{corollary} 
Since $G(\cdot)$ is convex in $\epsilon$, it should achieve its maximum at the endpoint of $[-\epsilon_1, \epsilon_2]$ or $\epsilon\in[-\epsilon_2, \epsilon_1].$ Therefore, following the above rounding steps~(\ref{eq:rounding-step}) and~(\ref{eq:rounding-result}), we have $G(\boldsymbol \delta, \tilde{\boldsymbol b}^\prime)\geq G(\boldsymbol \delta, \tilde{\boldsymbol b})$.
\end{corollary}
\begin{proof}
$G(\boldsymbol \delta, \tilde{\boldsymbol b}^\prime)\geq G(\boldsymbol \delta, \tilde{\boldsymbol b})$ follows directly from the convexity of $G(\cdot)$ in $\epsilon$ and the rounding steps in~(\ref{eq:rounding-step}) and~(\ref{eq:rounding-result}).
\end{proof}

\noindent{\textbf{Rounding Scheme:}}  Now for any solution $(\boldsymbol \delta, \tilde{\boldsymbol b})$ that satisfies the constraints in~(\ref{opt:max-gain-relaxed}) and~(\ref{opt:approx-max-gain-relaxed}), where $\tilde{\boldsymbol b}$ contains fractional terms.  There always exists a way to transfer mass between any two fractional variables $\tilde{b}_{k, j}$ and $\tilde{b}_{k, l}$ such that 
\begin{itemize}
\item (i) at least one of them becomes $0$ or $1$; 
\item (ii) the resultant solution $(\boldsymbol \delta, \tilde{\boldsymbol b}^\prime)$ is feasible, i.e., $(\boldsymbol \delta, \tilde{\boldsymbol b}^\prime)$ satisfy the constraints in~(\ref{opt:max-gain-relaxed}) and~(\ref{opt:approx-max-gain-relaxed}); 
\item (iii) the gain satisfies $G(\boldsymbol \delta, \tilde{\boldsymbol b}^\prime)\geq G(\boldsymbol \delta, \tilde{\boldsymbol b}).$
\end{itemize}

Then we can obtain an integral solution with the following iterative algorithm:
\begin{enumerate}
\item Given the optimal solution $(\boldsymbol \tau^*, \boldsymbol u^*)$ to~(\ref{opt:convextransformed}) , we first obtain the optimal solution $(\boldsymbol \delta^{**}, \tilde{\boldsymbol b}^{**})$ through the convexity mapping defined in~(\ref{newvariables}).
\item If there are fractional solutions in $\tilde{\boldsymbol b}^{**}$,  the number of fractional solutions must be at least two since the capacities are integer. W.l.o.g., consider two fractional solutions $\tilde{b}_{k, j}^{**}$ and $\tilde{b}_{k, l}^{**},$ for $j, l\in\{1, \cdots, h(k)\}$ and $j\neq l.$
\item Following the above properties (i) (ii) and (iii) to transform at least one of them into $0$ or $1$ and the resultant gain $G$ is increased. 
\item Repeat steps $2$ and $3$ until there are no fractional solutions in $\tilde{\boldsymbol b}^{**}.$ 
\end{enumerate}

Denote the resultant solution as $(\boldsymbol \delta^{**}, \tilde{\boldsymbol b}^{**\prime})$ which satisfies the constraints in~(\ref{opt:max-gain}). Note that each step can round at least one fractional solution to an integer one, the above iterative algorithm can terminate at most in $|\mathcal{K}|\times \sum_{k\in\mathcal{K}}|h(k)|$ steps. As each rounding step increases the gain,  we have 
\begin{align}
G(\boldsymbol \delta^{**}, \tilde{\boldsymbol b}^{**\prime})\geq G(\boldsymbol \delta^{**}, \tilde{\boldsymbol b}^{**})&\stackrel{(a)}{\geq}\left(1-\frac{1}{e}\right)G(\tilde{\boldsymbol \delta}^*, \tilde{\boldsymbol b}^*)\nonumber\displaybreak[0]\\
&\stackrel{(b)}{\geq}\left(1-\frac{1}{e}\right)G(\boldsymbol \delta^*, \boldsymbol b^*),
\end{align}
where (a) holds from Theorem~\ref{thm:relation-relaxed-opt-convexed-opt} and (b) holds from Theorem~\ref{thm:relation-oropt-relaxed-opt}. Therefore, we have obtained a $\left(1-1/e\right)$-approximation solution to the original optimization problem~(\ref{opt:max-gain}).

\section{Performance Evaluation}\label{sec:numerical}

We evaluate the performance of our proposed algorithm against benchmarks over synthetic {data}-based network topologies.

\begin{table}[ht]
	\vspace{-0.1in}
	\centering
	\caption{Characteristics of the Online Solvers}
	\vspace{-0.1in}
	\begin{tabular}{|l|p{6.0cm}|}
		\hline
		\textbf{Solver} & \textbf{Characteristics}  \\ \hline
		\textbf{Bonmin} \cite{bonami08} &  A deterministic approach based on Branch-and-Cut method that solves relaxation problem with Interior Point Optimization tool (IPOPT), as well as mixed integer problem with Coin or Branch and Cut (CBC).  \\ \hline
		\textbf{NOMAD} \cite{le11}  & A stochastic approach based on Mesh Adaptive Direct Search Algorithm (MADS) that guarantees local optimality. It can be used to solve non-convex MINLP.\\ \hline
		\textbf{GA} \cite{deb02} &  A meta-heuristic stochastic approach that can be tuned to solve global optimization problems.      \\ \hline
	\end{tabular}
	\vspace{-0.1in}
	\label{tab:Solvercharacteristics}
\end{table}

\subsection{Benchmarks}
To compare our proposed solution technique with existing ones, we solve the original  non-convex mixed integer non-linear optimization (MINLP) in~(\ref{gain}) using conventional online solvers, including  Bonmin \cite{bonami08}, NOMAD \cite{le11} and Genetic Algorithm (GA) \cite{deb02}, which have all been designed to solve classes of MINLP problems. The characteristics of these solvers are given in Table~\ref{tab:Solvercharacteristics}.

Note that GA is a stochastic approach whose performance greatly varies from one simulation run to other. In order to reduce the variance, we run the algorithm  10 times and provide the average, maximum and minimum time along with objective function value obtained using GA. For sake of comparison, we also ran our algorithm 10 times.   For our proposed algorithm, we use Algorithm~\ref{algo:alg} to solve the approximate relaxed convex problem and then use the rounding scheme discussed in Section~\ref{sec:rounding} to obtain a feasible solution to the original problem. We compare the performance of our proposed algorithm with these benchmarks with respect to average latency as well as the complexity (measured in time).

%
%\begin{table}[ht]
%	\centering
%	\caption{Parameters Used in Simulations}
%	\vspace{-0.1in}
%	\begin{tabular}{|c|c|}
%		\hline
%		\textbf{Parameter} & \textbf{Value}
%		\\ \hline
%		$y_k$	& 100           \\ \hline
%		$R_k$	& 1000             \\ \hline
%		$w_{ca}$	& 1.88 $\times$ 10$^{-6}$            \\ \hline
%		$T$	& 10s               \\ \hline
%		$\varepsilon_{vR}$ &  50  $\times$ 10$^{-9}$     \\ \hline
%		$\varepsilon_{vT}$ &   200  $\times$ 10$^{-9}$        \\ \hline
%		$\varepsilon_{cR}$ &  80  $\times$ 10$^{-9}$    \\ \hline
%		$S_v$ & 120  \\ \hline
%		$l$ &0.6 \\ \hline
%		$W$&$200$\\ \hline
%	\end{tabular}
%	\vspace{-0.3in}
%	\label{tab:algcompareparameters}
%\end{table}
%
%

\begin{table}[ht]
	\vspace{-0.1in}
	\centering
	\caption{Parameters Used in Simulations}
	\vspace{-0.1in}
	\begin{tabular}{|c|c||c|c|}
		\hline
		\textbf{Parameter} & \textbf{Value} &\textbf{Parameter} & \textbf{Value} \\ \hline
		$y_k$	& $100$ &  $\varepsilon_{vR}$ &  $50$  $\times$ $10^{-9}$ J           \\ \hline
		$R_k$	& $1000$    & $\varepsilon_{vT}$&   $200$  $\times$ $10^{-9}$ J              \\ \hline
		$w_{ca}$	& $1.88$ $\times$ $10^{-6}$&  $\varepsilon_{cR}$ &  $80$  $\times$ $10^{-9}$ J             \\ \hline
		$T$	& $10$s   &$l$ &$0.6$               \\ \hline
		$S_v$ &$120$ & $W$& $200$  \\ \hline
	\end{tabular}
	\vspace{-0.2in}
	\label{tab:algcompareparameters}
\end{table}

\subsection{Synthetic Evaluation}
\subsubsection{Simulation Setting}
We consider binary tree networks with $7, 15, 31$ and $63$ nodes, respectively.  We assume that each leaf node generates $y_k=100$ data items\footnote{Note that this can be equivalent taken as $100$ sensors generate data}, which will be requested $R_k=1000$ times during a time period $T=10s.$ $S_v=120$ is the storage capacity of each node. For simplicity, we assume that the latency along each edge in the network is identical and take $l=0.6.$ Our simulation parameters are provided in Table~\ref{tab:algcompareparameters}, which are typical values used in the literature \cite{nazemi16,heinzelman00,ye02}.  We implement Bonmin, NOMAD and Algorithm~\ref{algo:alg} in Matlab using OPTI-Toolbox and Matlab's built-in GA algorithm on a Windows $7$ $64$ bits,  $3.40$ GHz Intel Core-i$7$ Processor with a $16$ GB memory.% $2133$ MHz LPDDR$3$ memory.}

\begin{table*}[]
	\centering
	\caption{Comparison Among Selected Algorithms Using Synthetic Data for Various Network Topologies}
	\vspace{-0.1in}
	\label{tab:algresults}
	\begin{tabular}{|l|l|l|l|l|l|l|l|l|}
		\hline
		\multicolumn{1}{|c|}{\multirow{2}{*}{\textbf{Nodes}}} & \multicolumn{2}{c|}{\textbf{Proposed}}                                             & \multicolumn{2}{c|}{\textbf{GA}}                                              & \multicolumn{2}{c|}{\textbf{Nomad}}                                           & \multicolumn{2}{c|}{\textbf{Bonmin}}                                          \\ \cline{2-9} 
		\multicolumn{1}{|c|}{}                                & \multicolumn{1}{c|}{\textbf{Obj. Value}} & \multicolumn{1}{c|}{\textbf{Time(s)}} & \multicolumn{1}{c|}{\textbf{Obj. Value}} & \multicolumn{1}{c|}{\textbf{Time(s)}} & \multicolumn{1}{c|}{\textbf{Obj. Value}} & \multicolumn{1}{c|}{\textbf{Time(s)}} & \multicolumn{1}{c|}{\textbf{Obj. Value}} & \multicolumn{1}{c|}{\textbf{Time(s)}} \\ \hline
		$\textbf{7}$                                            & $480000$                                   & $3.30$                               & $479820$                                   & $273.29$                             & $\text{Infeasible}$                                      & $9.45$                               & $\text{Infeasible}$                                      & $1.01$                               \\ \hline
		$\textbf{15}$                                           & $1440000$                                  & $6.33$                              & $1440000$                                  & $15.12$                              & $1439900$                                  & $16.18$                              & $\text{Infeasible}$                                       & $>4000$                   \\ \hline
		$\textbf{31}$                                           & $3840000$                                  & $29.28$                             & $3839000$                                  & $3501.10$                            & $\text{Non-Convergence}$                                       & $98.90$                              & $\text{Non-Convergence}$                                          & $1232.31$                            \\ \hline
		$\textbf{63}$                                           & $9599900
		$                                  & $538.17$                            & $8792100$                                  & $158.56$                             & $\text{Non-Convergence}$                                         & $966.16$                            & $\text{Non-Convergence}$                                          & $2.04 $                              \\ \hline
	\end{tabular}
	\vspace{-0.10in}
\end{table*}

\begin{table*}[]
	\centering
	\caption{Detailed Results for Our Proposed Algorithm}
	        \vspace{-0.1in}
		\begin{tabular}{|l|l|l|l|l|l|l|l|l|l|}
		\hline
		\multirow{2}{*}{\textbf{Node}} & \multicolumn{3}{c|}{\textbf{Time}}                                                                            & \multicolumn{3}{c|}{\textbf{Obj. Value}}                                                                      & \multicolumn{3}{c|}{\textbf{Iterations}}                                                                      \\ \cline{2-10} 
		& \multicolumn{1}{c|}{\textbf{Max.}} & \multicolumn{1}{c|}{\textbf{Min.}} & \multicolumn{1}{c|}{\textbf{Average}} & \multicolumn{1}{c|}{\textbf{Max.}} & \multicolumn{1}{c|}{\textbf{Min.}} & \multicolumn{1}{c|}{\textbf{Average}} & \multicolumn{1}{c|}{\textbf{Max.}} & \multicolumn{1}{c|}{\textbf{Min.}} & \multicolumn{1}{c|}{\textbf{Average}} \\ \hline
		$\textbf{7} $                    & $3.5848$                            & $3.12$                              & $3.30 $                                 & $480000$                            & $480000$                            & $480000$                                & $4$                                 & $4$                                 & $4$                                     \\ \hline
		$\textbf{15} $                   & $7.23$                              & $6.00$                              & $6.33$                                  & $1440000$                           & $1440000$                           & $1440000$                               & 2                                 & $2$                                 & $2$                                     \\ \hline
		$\textbf{31}$                    & $30.99$                             & $28.57$                             & $29.28 $                                & $3840000$                           & $3839900$                           & $3840000$                               & $2$                                 & $2$                                 & $2$                                     \\ \hline
		\textbf{63}                    & $553.94 $                           & $531.61$                            & $538.17$                                & $9600000$                           & $9599900$                           & $9599900$                               &$ 3$                                 & $3$                                 & $3$                                     \\ \hline
	%	$\textbf{63 (trace)} $           & $84.1899$                           & $79.279$                            & $82.24$                                 & $672260$                            &$ 672220$                            & $672250$                                & $2$                                 & $2$                                 & $2$                                     \\ \hline
	\end{tabular}
\label{tab:ourdetails}
\vspace{-0.10in}
\end{table*}

\begin{table*}[ht]
	\centering
	\caption{Robustness of GA Algorithm}
	\vspace{-0.1in}
	\label{tab:garesults}
\begin{tabular}{|l|l|l|l|l|l|l|l|}
	\hline
	\multirow{2}{*}{\textbf{Node}} & \multicolumn{3}{c|}{\textbf{Time (s)}} & \multicolumn{3}{c|}{\textbf{Objective Value}} & \multicolumn{1}{c|}{\multirow{2}{*}{\textbf{Convergence (\%)}}} \\ \cline{2-7}
	& \textbf{Max.}         & \textbf{Min.}         & \textbf{Average}    & \textbf{Max.}         & \textbf{Min.}         & \textbf{Average}     & \multicolumn{1}{c|}{}                                           \\ \hline
$	\textbf{7} $                    & $369.43$      & $161.76 $     & $273.29$     & $479880$       & $479750 $     & $479820 $     & $100 $                                                            \\ \hline
$	\textbf{15} $                   & $18.15 $      & $12.56 $      &$ 15.12$      & $1440000$      & $1440000 $    & $1440000$     & $100  $                                                           \\ \hline
$	\textbf{31}$                    & $4446.70$     & $2552.40 $    & $3501.10$    & $3839100$      & $3838900 $    & $3839000  $   &$ 100$                                                             \\ \hline
$	\textbf{63} $                   & $413.28 $    & $24.41$       & $158.56 $    & $9599100$      & $8041500$     & $8792100$     & $40 $                                                             \\ \hline
%$	\textbf{63 (trace)}  $          & $187.73 $     & $133.14$      & $153.52$     & $671997 $      & $665390 $     & $670310$      & $100 $                                                            \\ \hline
\end{tabular}
\vspace{-0.10in}
\end{table*}

\subsubsection{Evaluation Results}
The performance of these algorithms with respect to the obtained value of the objective function and the time needed to obtain it, are given in Table~\ref{tab:algresults}.%, where $\text{Inf.}$ means the value is infeasible and $\text{N.C.}$ means non convergence.  

On the one hand, we observe that neither Bonmin or NOMAD provide any feasible solution with the constraint in~(\ref{eq:constraint4}).  We then further relax this constraint for Bonmin and NOMAD. Hence, the results provided in Table~\ref{tab:algresults} for Bonmin and NOMAD are solved without constraint~(\ref{eq:constraint4}).   Again, we notice that even after relaxing the constraint, Bonmin and NOMAD still exhibit a poor performance, i.e. they either providing an infeasible solution or not converging to a feasible solution. This is mainly due to the hardness of the original non-convex MINLP~(\ref{gain}).  Hence, it is important to provide an efficient approximation algorithm to solve it. 

On the other hand, we observe that both our proposed algorithm and GA provide encouraging results.  We run both GA and our algorithm $10$ times and report their average as the obtained solutions and run time in Table~\ref{tab:algresults}.  Tables~\ref{tab:ourdetails} and~\ref{tab:garesults} provide detailed results for both algorithms, respectively.  It is clear that our proposed algorithm significantly outperforms these conventional online solvers both in terms of run time and the obtained objective function value.  

In particular, for the $63$ nodes network, GA provides an output solution faster than ours. However, GA is not robust and reliable for larger networks. We characterize the robustness of GA, as shown in Table~\ref{tab:garesults}, where the maximal (Max.), minimal (Min.) and average values of the objective function are presented as well as the corresponding time to obtain them for $7, 15, 31$ and $63$ nodes binary tree networks, respectively. We notice that for the $63$ nodes network, only $4$ out of $10$ runs converge to a feasible solution using GA. Therefore, GA cannot always guarantee a feasible solution though it may complete in less time. Table~\ref{tab:ourdetails} provides a detailed overview of our algorithm. The maximal, minimal and average values in terms of time, obtained solution and number of iterations are given. Our master-slave algorithm can converge to a solution in small number of iterations. 

Also note that our proposed approach always achieves a feasible solution within the $(1-1/e)$ approximation of the optimal solution\footnote{Note that the original optimization problem~(\ref{gain}) is non-convex MINLP, which {is} NP-hard.  Bonmin, NOMAD and GA all claim to solve MINLP with a $\epsilon$-optimal solution.  {However, GA and NOMAD are stochastic approaches, they cannot guarantee $\epsilon$-global optimality.}  Hence, we compare our solution with these of Bonmin, NOMAD and GA to verify the approximation ratio.}\footnote{$\epsilon$-global optimality means that the obtained solution is within $\epsilon$ tolerance of the global optimal solution.}.  Therefore, our proposed algorithm can efficiently solve the problem, i.e., providing a feasible solution in a reasonable time and is robust to network topologies changes.

We also characterize the impact of number of requests on the caching and compression gain, shown in Figure~\ref{fig:objvsr}.  We observe that as the number of requests increase, the gain increases, as reflected in the objective function~(\ref{gain}). Since the objective function~(\ref{gain}) is monotonically increasing in the number of requests $R_k$ for all $k\in\mathcal{K}$ provided that $\boldsymbol\delta$ and $\boldsymbol b$ are fixed.

\begin{comment}
\subsection{Trace-based Evaluation}
We extract data and network topologies from real traces.  We use the trace-files \cite{queensu-crowd_temperature-20151120}, which are the temperature measurements obtained by $289$ different taxis across Rome and reported  to a central entity at random time intervals.  We treat each data point as a single data bit, combine all data bits of a single taxi, and use it as $y_k$ for that particular data source.  We randomly select $63$ different taxis and use them as the data sources. Thus the network consisted of $63$ leaf nodes and a root node.  All other parameters used in the evaluation are the same as those in Table~\ref{tab:algcompareparameters}.

Again, Bonmin and NOMAD cannot provide any feasible solution. The performance of GA under this trace-based simulation is given in the last row of Table~\ref{tab:garesults}.  We can see that the performance is not quite stable with a large variance on the value of the objective function and the running time. 
Our proposed algorithm, on the other hand, converges to a feasible solution in $29.15$s and obtains an objective function value of  $672260.31$.  
\end{comment}

\begin{figure}
\centering
\includegraphics[width=1\linewidth]{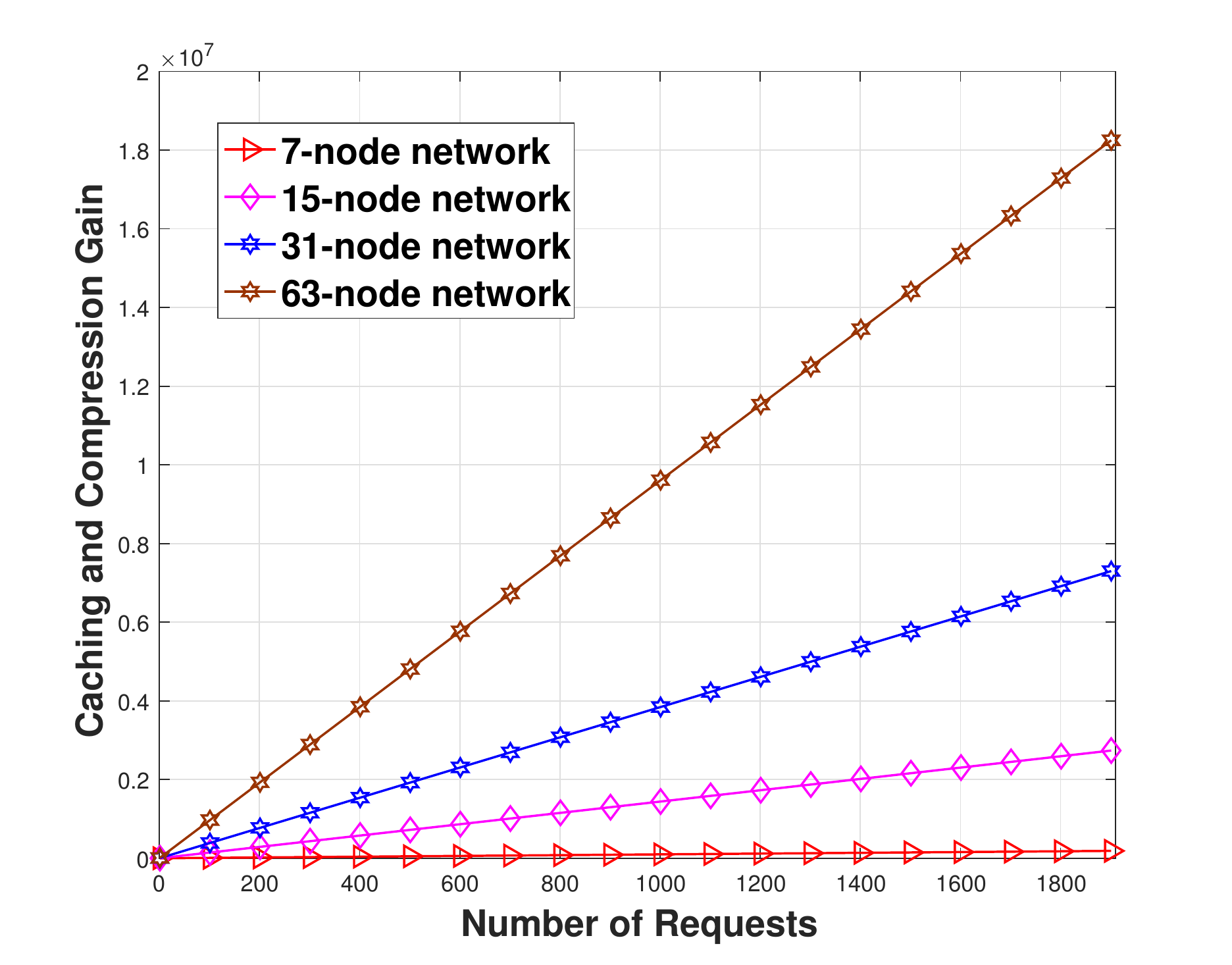}
\vspace{-0.3in}
\caption{Impact of number of requests on the performance.}
\label{fig:objvsr}
\vspace{-0.25in}
\end{figure}

\begin{remark}
Throughout the evaluations, we notice that the compression ratio at leaf node is much smaller than the ratio at root node.  For example, in the $63$ nodes network, the compression ratio{\footnote{{Defined as the ratio of the volume of the output data to that of the input data at a node. The higher the compression ratio is, the lower is the {data compression}.}}} at leaf node is $0.01$ while it is $0.37$ at the root node.  This captures the tradeoff between the costs of compression, communication and caching in our optimization framework.   Similar observations can be made in other networks and hence are omitted here. 
\end{remark}

%\red{It is worth mentioning that throughout the evaluations, most of the compression was done at nodes down in the tree. The highest reduction rate (lowest compression) was observed at the root node.} 

\subsubsection{Heterogeneous Networks}\label{sec:numerical2}

%We evaluate the performance of our proposed algorithm against benchmarks over synthetic-based network topologies.

{In the previous section, we consider binary tree networks under homogeneous settings, i.e., the value of different parameters are identical for all nodes in the network, as given in Table~\ref{tab:algcompareparameters}.  In this section, we generalize the simulation setting from two perspectives: (i) First, we consider heterogeneous parameter values across the network.  For example, for the node cache capacity $S_v,$ we assume that $S_v=100{+rand(1,20)}$, where $rand(i,j)$ assigns a random number between $i$ and $j$.  Similarly, we assign a random number to $\varepsilon_{vR},$ $\varepsilon_{vT}$ and $\varepsilon_{cR}$ on each node;  (ii) Second, instead of considering binary tree, we consider more general network topologies with $7, 15, 31$ and $67$ nodes, as shown in Figure~\ref{fig:networks}.}

\begin{figure}
\centering
\includegraphics[width=1\linewidth]{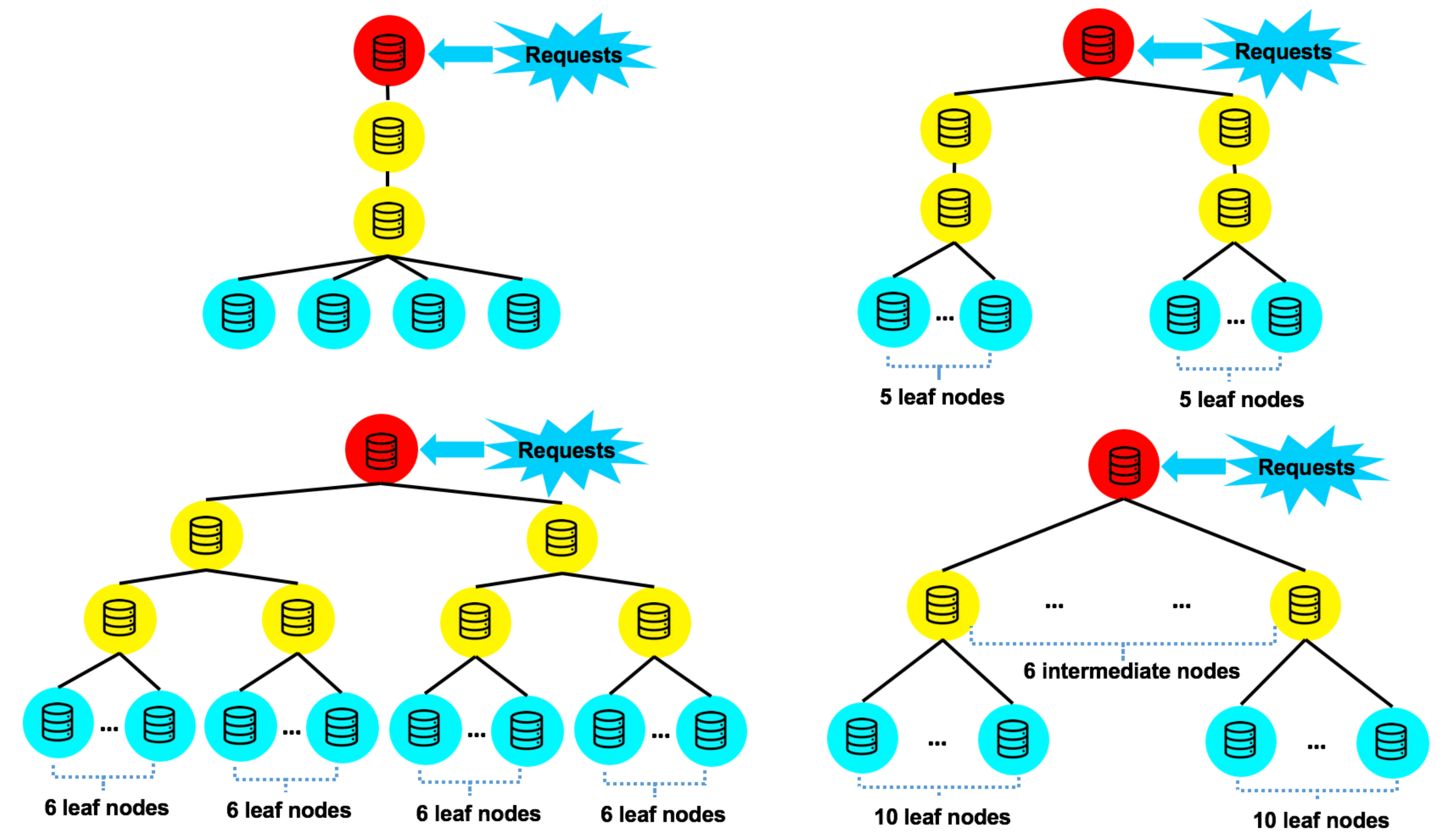}
\vspace{-0.3in}
\caption{Heterogeneous Tree Networks used in Simulations}
	\label{fig:networks}
\vspace{-0.2in}
\end{figure}

{The performance of these algorithms with respect to the obtained value of the objective function and the time needed to obtain it, are given in Table~\ref{tab:algresults-new}.  Again, we observe that neither Bonmin nor NOMAD can effectively solve the original problem in ~(\ref{gain}), which shows the hardness of the problem. Hence, it is important to provide an efficient approximation algorithm to solve it. }

{Similarly, we also observe that both our proposed algorithm and GA provide encouraging results.  We run both GA and our algorithm $10$ times and report their average as the obtained solutions and run time in {Table~\ref{tab:algresults-new}}.  %We also obtain detailed results for both algorithms, similar trends exists in Tables~\ref{tab:ourdetails} and~\ref{tab:garesults}, respectively.  Hence are omitted here due to space constraints, which are available in \cite{jianfaheem18}. %
	Tables~\ref{tab:ourdetailshet} and~\ref{tab:garesultshet} provide detailed results for both algorithms, respectively in the heterogeneous setting.  
 It is clear that our proposed algorithm significantly outperforms these conventional online solvers both in terms of run time and the obtained objective function value.  Furthermore, again we notice that GA cannot always guarantee a feasible solution. }

{We also characterize the impact of number of requests on the caching and compression gain as shown in Figure~\ref{fig:objvsr2}.  Similar to Figure~\ref{fig:objvsr}, we observe that as the number of requests increase, the gain increases. }

\begin{table*}[]
	%\color{blue}
	\centering
	\caption{{Comparison Among Selected Algorithms Using Synthetic Data for Various Network Topologies II}}
	\vspace{-0.1in}
	\label{tab:algresults-new}
	\taburulecolor{blue}
	\begin{tabular}{|l|l|l|l|l|l|l|l|l|}
		\hline
		\multicolumn{1}{|c|}{\multirow{2}{*}{\textbf{Nodes}}} & \multicolumn{2}{c|}{\textbf{Proposed}}                                             & \multicolumn{2}{c|}{\textbf{GA}}                                              & \multicolumn{2}{c|}{\textbf{Nomad}}                                           & \multicolumn{2}{c|}{\textbf{Bonmin}}                                          \\ \cline{2-9} 
		\multicolumn{1}{|c|}{}                                & \multicolumn{1}{c|}{\textbf{Obj. Value}} & \multicolumn{1}{c|}{\textbf{Time(s)}} & \multicolumn{1}{c|}{\textbf{Obj. Value}} & \multicolumn{1}{c|}{\textbf{Time(s)}} & \multicolumn{1}{c|}{\textbf{Obj. Value}} & \multicolumn{1}{c|}{\textbf{Time(s)}} & \multicolumn{1}{c|}{\textbf{Obj. Value}} & \multicolumn{1}{c|}{\textbf{Time(s)}} \\ \hline
		$\textbf{7}$   & $720000$ & $5.5261	$ & $720000$  & $6.14
		$  & $720000$    & $78.99$  & $\text{Infeasible}$                                      & $45.37
		$                               \\ \hline
		$\textbf{15}$   & $1799900$    & $5.87$  & $1800000	$     & $20.08$                              & $\text{Non-Convergence}$      & $37.27$                              & $\text{Infeasible}$& $1151.05
		$                   \\ \hline
		$\textbf{31}$  & $4319500$ & $33.81$ & $4318700$ & $66.68$                            & $\text{Non-Convergence}$ & $179.21$   &$\text{Non-Convergence}$                                          & $32293.37$                            \\ \hline
		$\textbf{67}$  & $7197900$ & $115.17$  & $6531200$  & $399.39$    & $\text{Non-Convergence}$                                         & $1037$                            & $\text{Non-Convergence}$                                          & $>40000$                              \\ \hline
	\end{tabular}
		\vspace{-0.15in}
\end{table*}

\begin{table*}[]
	\centering
	\caption{Detailed Results for Our Proposed Algorithm}
	\vspace{-0.1in}
	\begin{tabular}{|l|l|l|l|l|l|l|l|l|l|}
		\hline
		\multirow{2}{*}{\textbf{Node}} & \multicolumn{3}{c|}{\textbf{Time}}                                                                            & \multicolumn{3}{c|}{\textbf{Objective Value}}                                                                      & \multicolumn{3}{c|}{\textbf{Iterations}}                                                                      \\ \cline{2-10} 
		& \multicolumn{1}{c|}{\textbf{Max.}} & \multicolumn{1}{c|}{\textbf{Min.}} & \multicolumn{1}{c|}{\textbf{Average}} & \multicolumn{1}{c|}{\textbf{Max.}} & \multicolumn{1}{c|}{\textbf{Min.}} & \multicolumn{1}{c|}{\textbf{Average}} & \multicolumn{1}{c|}{\textbf{Max.}} & \multicolumn{1}{c|}{\textbf{Min.}} & \multicolumn{1}{c|}{\textbf{Average}} \\ \hline
		\textbf{7}  & $5.87$   & $5.34$  & $5.52 $   & $720000$ & $720000$  & $720000$ & $3 $   & $3 $ & $3 $  \\ \hline
		\textbf{15}                    & $6.39$  & $5.50$  & $5.87$ & $1799900$   & $1799900$                      & $1799900$   & $2$ & $2$ & $2$ \\ \hline
		\textbf{31} & $36.44$  & $32.94$ & $33.81$  & $4319500$ & $4319500$ &$4319500 $   & $2$ & $2$ & $2$ \\ \hline
		\textbf{67}                    & $120.92$   & $108.85$  & $115.17$  & $7197900$  & $7197900$ & $7197900$  & $2$ & $2$ & $2$ \\ \hline
	\end{tabular}
	\label{tab:ourdetailshet}
	\vspace{-0.10in}
\end{table*}
\begin{table*}[ht]
	\centering
	\caption{Robustness of GA Algorithm}
	\vspace{-0.1in}
	\label{tab:garesultshet}
	\begin{tabular}{|l|l|l|l|l|l|l|l|}
		\hline
		\multirow{2}{*}{\textbf{Node}} & \multicolumn{3}{c|}{\textbf{Time (s)}} & \multicolumn{3}{c|}{\textbf{Objective Value}} & \multicolumn{1}{c|}{\multirow{2}{*}{\textbf{Convergence (\%)}}} \\ \cline{2-7}
		& \textbf{Max.}         & \textbf{Min.}         & \textbf{Average}    & \textbf{Max.}         & \textbf{Min.}         & \textbf{Average}     & \multicolumn{1}{c|}{}                                           \\ \hline
		\textbf{7}                     & $7.26$          & $5.09$   & $6.14$ & $720000$  & $720000$ & $720000$  & $100$   \\ \hline
		\textbf{15}                    & $30.13$         & $15.54 $        & $20.08 $           &$ 1800000$       & $1799800$       & $1800000  $        & $60 $                       \\ \hline
		\textbf{31}                    & $100.93$        & $48.02$         & $66.68  $          & $4319000$       & $4318300$       & $4318700$          & $70$                        \\ \hline
		\textbf{67}                    & $2363.60$       & $23.71$         & $399.39$           & $7194400$       & $4758800$       & $6531200$          & $40$                        \\ \hline
	\end{tabular}
	\vspace{-0.10in}
\end{table*}

\begin{figure}
	\centering
	\includegraphics[width=1\linewidth]{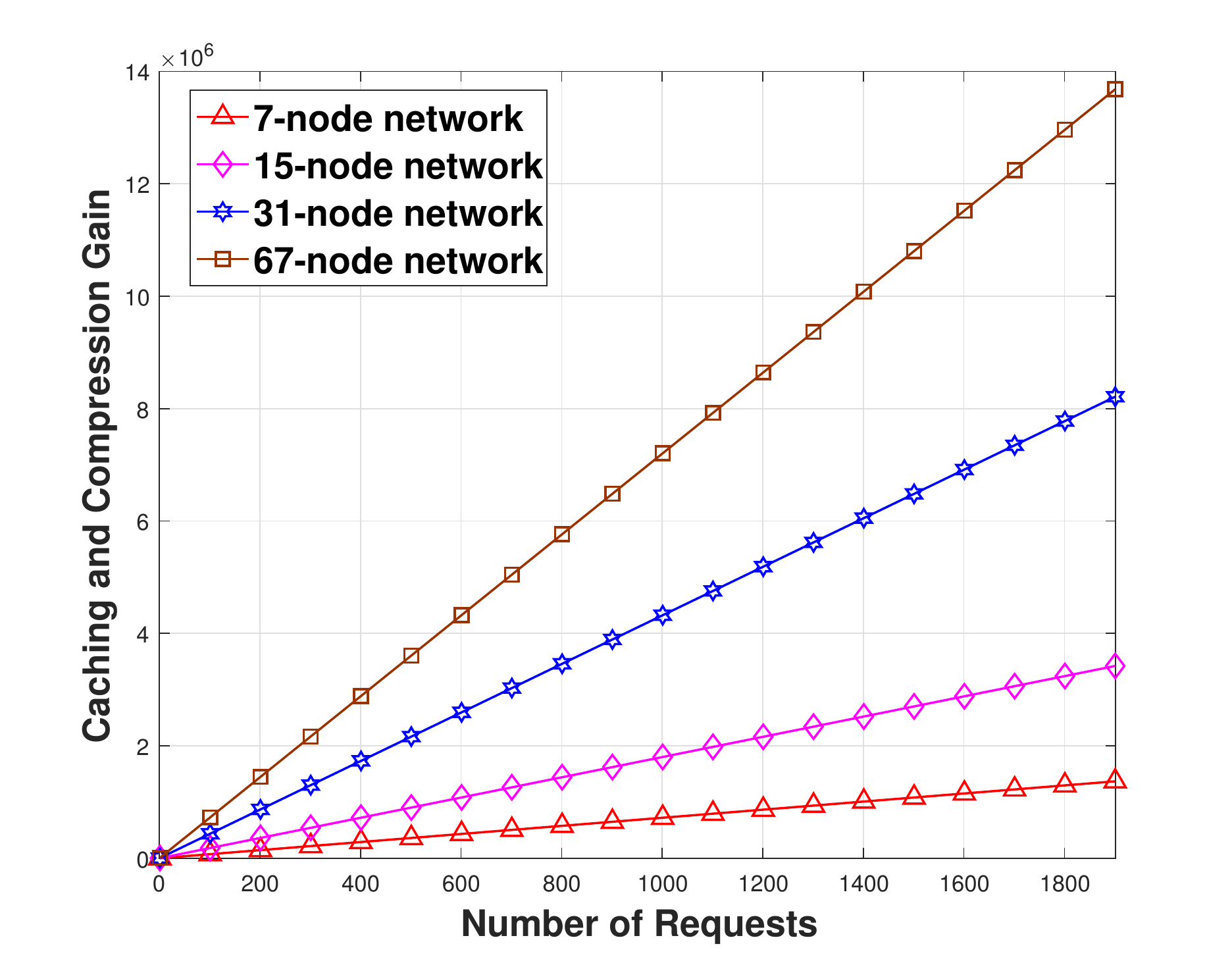}%objvsR2}
	\vspace{-0.2in}
	\caption{Impact of number of requests on the performance for the heterogeneous network.}
	\label{fig:objvsr2}
	\vspace{-0.3in}
\end{figure}

\section{Conclusion}\label{sec:con}

We considered the problem of optimally compressing and caching data across a communication network.  In order to achieve this goal, we formulated an optimization problem to minimize the total latency, which is equivalent to maximizing compression and caching gain, under energy constraint. This problem is NP-hard in general. We then proposed an efficient approximation algorithm that can achieve a $(1-1/e)$ approximation solution to the optimum in strongly polynomial time. Finally, we evaluated the performance of our proposed algorithm through extensive synthetic simulations, and made a comparison with benchmarks. We observed that our proposed algorithm can achieve near-optimal solution and outperform the benchmarks.

\vspace{-0.1in}
\section*{Acknowledgments}\label{sec:ack}
This work was supported by the U.S. Army Research Laboratory and the U.K. Ministry of Defence under Agreement Number W911NF-16-3-0001.  The views and conclusions contained in this document are those of the authors and should not be interpreted as representing the official policies, either expressed or implied, of the U.S. Army Research Laboratory, the U.S. Government, the U.K. Ministry of Defence or the U.K. Government. The U.S. and U.K. Governments are authorized to reproduce and distribute reprints for Government purposes notwithstanding any copy-right notation hereon. 
Faheem Zafari also acknowledges the financial support by EPSRC Centre for Doctoral Training in High Performance Embedded and Distributed Systems  (HiPEDS, Grant Reference EP/L016796/1), and Department of Electrical and Electronics Engineering, Imperial College London.

\vspace{-0.1in}
\bibliographystyle{ACM-Reference-Format}
\bibliography{refs} 

\section{Appendix}\label{sec:app}
Equation~\eqref{eq:servingcost} captures one-time\footnote{During every time period  $T$,  data is always pushed towards the sink upon the first request.} energy cost of receiving, compressing and transmitting data $y_k$ from leaf node (level $h(k)$) to the sink node (level $0$).  The amount of data received by any node at level $i$ from leaf node $k$ is $y_k\prod_{m=i+1}^{h(k)}\delta_{k,m}$ due to the compression from level $h(k)$ to $i+1.$  The term $f(\delta_{k,i})$ captures the reception, transmission and compression energy cost for node at level $i$ along the path from leaf node $k$ to the sink node.  
 
 For Equation~\eqref{eq:cachretrievecost}, note that the remaining $(R_k-1)$ requests are either served by the leaf node or a cached copy of data $y_k$ at level $i$ for $i=1,\cdots, h(k)$ \cite{faheemjian17}.  W.l.o.g., we consider node $v_{k, i}$  at level $i$.  If data $y_k$ is not cached from $v_{k, i}$ up to the sink node (level $0)$, i.e., $b_{k, j}=0$ for $j=0, \cdots, i,$  the cost is {incurred} due to  receiving, transmitting and compressing the data $(R_k-1)$ times, which is captured by the first term in Equation~\eqref{eq:cachretrievecost}, the second term is $0$.  Otherwise, the $(R_k-1)$ requests are served by the cached copy at $v_{k, i}$,  the corresponding caching and transmission cost serving from $v_{k, i}$ are captured by the second term in Equation~\eqref{eq:cachretrievecost}, and the corresponding reception, transmission and compression cost from $v_{k, i-1}$ up to to sink node is captured by the first term.  Note that the first time cost of reception, transmission and compression the data from leaf node to $v_{k, i}$ is {already} captured by Equation~\eqref{eq:servingcost}.  
 
{For sake of completeness, we restate a simple but illustrative example from \cite{faheemjian17} to explain the above equations.}
\begin{example}\label{exm1}
	{We consider a network with one leaf node and one sink node, i.e., $k=1$ and $h(k)=1.$ }

	{Then the cost in Equation~(\ref{eq:servingcost}) becomes $E_1^C=y_1f(\delta_{1, 0})\delta_{1,1}+y_1f(\delta_{1,1}),$ where the first and second terms capture the reception, transmission and compression cost for data $y_1$ at sink node and the leaf node, respectively. }
	
	{The cost in Equation~(\ref{eq:cachretrievecost}) is $E_1^R=$
		\begin{small}
			\begin{align*}
			&\underbrace{y_1(R_1-1)\left[f(\delta_{1, 0})\delta_{1,1}(1-b_{1,0})+\delta_{1,0} \delta_{1,1}b_{1,0}\left(\frac{w_{ca}T}{R_1-1}+\varepsilon_{1T}\right)\right]}_{\text{Term $1$}}\nonumber\displaybreak[0]\\
			&{+\underbrace{y_1(R_1-1)\left[f(\delta_{1, 1})(1-b_{1,0}-b_{1,1})+ \delta_{1,1}b_{1,1}\left(\frac{w_{ca}T}{R_1-1}+\varepsilon_{1T}\right)\right]}_{\text{Term $2$}}},
			\end{align*}
		\end{small}
		where $\text{Term $1$}$ and $\text{Term $2$}$ capture the costs at sink node and leaf node, respectively.   To be more specific, there are three cases: (i) data $y_1$ is cached at sink node $0$, i.e., $b_{1, 0}=1$ and $b_{1, 1}=0$ (since we only cache one copy);  (ii) data $y_1$ is cached at leaf node $1$, i.e., $b_{1, 0}=0$ and $b_{1, 1}=1$; and (iii) data $y_1$ is not cached, i.e., $b_{1, 0}=b_{1, 1}=0$.  We consider these three cases in the following. }
	
	{Case (i), i.e., $b_{1, 0}=1$ and $b_{1, 1}=0$,  $\text{Term $2$}$ becomes $0$ and $\text{Term $1$}$ reduces to $y_1(R_1-1)\delta_{1,0} \delta_{1,1}b_{1,0}(\frac{w_{ca}T}{R_1-1} +\varepsilon_{1T})$ since all the $(R_1-1)$ requests are served from sink node. {This indicates that the total energy cost is due to caching the data for time period $T$ and transmitting it $(R_k-1)$ times from the sink node to users that request it. }} 
	
	{Case (ii), i.e., $b_{1, 0}=0$ and $b_{1, 1}=1$,  $\text{Term $1$}$ becomes $y_1(R_1-1)f(\delta_{1, 0})\delta_{1,1}$, which captures the reception, transmission and compression costs at sink node $0$ for serving the $(R_1-1)$ requests.  $\text{Term $2$}$ becomes $y_1(R_1-1) \delta_{1,1}b_{1,1}\left(\frac{w_{ca}T}{R_1-1}+\varepsilon_{1T}\right)$, which captures the cost {of caching data at the leaf node and transmitting the data $(R_k-1)$ times from the cached copy to the sink node} .  The sum of them is the total cost to serve $(R_1-1)$ requests. }
	
	{Case (iii), i.e., $b_{1, 0}=b_{1, 1}=0$, $E_1^R=y_1(R_1-1)f(\delta_{1, 0})\delta_{1,1}+y_1(R_1-1)f(\delta_{1, 1})$, which captures the reception, transmission and compression costs at sink node $0$ and leaf node $1$ for serving the $(R_1-1)$ requests since there is no cached copy in the network.  }
\end{example}

\end{document}